\documentclass[acmsmall,screen,nonacm]{acmart}

\usepackage{tikz}
\usetikzlibrary{arrows.meta,fit,calc,backgrounds}
\usepackage{amsmath}
\usepackage{comment}
\usepackage{xspace}
\usepackage{enumitem} %
\usepackage{csquotes} %
\usepackage{booktabs}
\usepackage{makecell}
\usepackage[font=small,skip=5pt]{caption} %
\usepackage[font=small,skip=5pt]{subcaption}
\usepackage{algorithm}
\usepackage{algpseudocode}
\usepackage{tabularx}
\usepackage{multirow}
\usepackage{soul}
\usepackage{hyperref}
\usepackage{wrapfig}
\usepackage{xcolor}
\usepackage[breakable, skins]{tcolorbox}
\newcommand{\code}[1]{%
  \tikz[baseline=(X.base)]\node[fill=gray!12,rounded corners=3pt,inner xsep=2pt,inner ysep=2pt,text height=0.6em,text depth=0.1em,anchor=base] (X) {\texttt{\small{}#1}};%
}

\tcbset{arc=0mm,size=fbox,attach title to upper={\ ---\ },coltitle=black,fonttitle=\bfseries}

\DeclareCaptionLabelFormat{custom}{#2)}
\newcommand{\tool}{\texttt{SPARC}\xspace}
\newcommand{\name}{\textsc{SPARC-Shield}\xspace}
\definecolor{cutcolor}{HTML}{E07000}

\definecolor{phaseone}{HTML}{083b8c}
\definecolor{phasetwo}{HTML}{008a0e}
\definecolor{phasethree}{HTML}{e81313}
\definecolor{phasefour}{HTML}{ba23f6}
\definecolor{phasefive}{HTML}{82755b}

\newlength{\circlebaseline}
\DeclareRobustCommand*\circleone{\tikz[baseline=\circlebaseline]{%
    \node[shape=circle,draw=phaseone,fill=phaseone,text=white,inner sep=1pt,font=\scriptsize\bfseries] (char) {1};}\xspace}
\DeclareRobustCommand*\circletwo{\tikz[baseline=\circlebaseline]{%
    \node[shape=circle,draw=phasetwo,fill=phasetwo,text=white,inner sep=1pt,font=\scriptsize\bfseries] (char) {2};}\xspace}
\DeclareRobustCommand*\circlethree{\tikz[baseline=\circlebaseline]{%
    \node[shape=circle,draw=phasethree,fill=phasethree,text=white,inner sep=1pt,font=\scriptsize\bfseries] (char) {3};}\xspace}
\DeclareRobustCommand*\circlefour{\tikz[baseline=\circlebaseline]{%
    \node[shape=circle,draw=phasefour,fill=phasefour,text=white,inner sep=1pt,font=\scriptsize\bfseries] (char) {4};}\xspace}
\DeclareRobustCommand*\circlefive{\tikz[baseline=\circlebaseline]{%
    \node[shape=circle,draw=phasefive,fill=phasefive,text=white,inner sep=1pt,font=\scriptsize\bfseries] (char) {5};}\xspace}

\definecolor{stepA}{HTML}{618B55}
\definecolor{stepB}{HTML}{BF5952}
\definecolor{stepC}{HTML}{8861A2}
\definecolor{stepD}{HTML}{51668A}
\definecolor{stepE}{HTML}{7A7A00}
\DeclareRobustCommand*\stepbox[2]{\tikz[baseline=\circlebaseline]{%
    \node[shape=rectangle,rounded corners=0.8pt,draw=#1,fill=#1,text=white,inner sep=0pt,minimum size=1.15em,font=\scriptsize\bfseries] (char) {#2};}}
\DeclareRobustCommand*\stepA{\stepbox{stepA}{A}\xspace}
\DeclareRobustCommand*\stepB{\stepbox{stepB}{B}\xspace}
\DeclareRobustCommand*\stepC{\stepbox{stepC}{C}\xspace}
\DeclareRobustCommand*\stepD{\stepbox{stepD}{D}\xspace}
\DeclareRobustCommand*\stepE{\stepbox{stepE}{E}\xspace}
\definecolor{rqframe}{HTML}{083b8c}
\DeclareRobustCommand*\rqbox[1]{\smash{\tikz[baseline=(char.base)]{%
    \node[shape=rectangle,rounded corners=0.8pt,draw=rqframe,line width=0.5pt,inner xsep=2pt,inner ysep=1.5pt,text height=1.6ex,text depth=0.3ex,text=rqframe,font=\scriptsize\bfseries] (char) {RQ#1};}}}

\setcopyright{none}
\renewcommand\footnotetextcopyrightpermission[1]{} 

\begin{document}
\title{\tool: Substation-to-Peering Assessment of Resilience and Cascades}

\author{Tanya Shreedhar}
\email{t.shreedhar@tudelft.nl}
\affiliation{%
  \institution{TU Delft}
  \country{The Netherlands}
}

\author{Kevin Vermeulen}
\email{kevinmylinh.vermeulen@gmail.com}
\affiliation{%
  \institution{LIX, CNRS, Ecole Polytechnique}
  \country{France}
}

\author{Amreesh Phokeer}
\email{phokeer@isoc.org}
\affiliation{%
  \institution{Internet Society}
  \country{USA}
}

\author{Italo Cunha}
\email{cunha@dcc.ufmg.br}
\affiliation{%
  \institution{Universidade Federal de Minas Gerais}
  \country{Brazil}
}

\author{Fernando Kuipers}
\email{f.a.kuipers@tudelft.nl}
\affiliation{%
  \institution{Delft University of Technology (TU Delft)}
  \country{The Netherlands}
}

\renewcommand{\shortauthors}{Shreedhar et al.}

\begin{abstract}
Internet Exchange Points (IXPs), where networks exchange a large share of Internet traffic, operate from colocation facilities that draw power from the regional grid and store a limited supply of backup fuel.
In April 2025, the Iberian blackout kept the Spanish grid down for 16~hours, longer than the 12~hours of fuel that the industry standard requires, and traffic at major Iberian IXPs fell by up to 70\%.
Yet operators cannot tell which facilities such an event cuts off, whether the outage outlasts their backup, or where the traffic of a failed IXP goes.
Prior work leaves these questions open, as it maps the grid dependency of the Internet only at the scale of whole synchronous grids or detects outages after they happen.
We present \tool, to our knowledge the first study that maps the interdependence of the power grid and Internet infrastructure from the substation to the IXP and to its traffic.
\tool extracts the high-voltage substations from open data, checks them against the registries of transmission operators, and links each facility to the substations that likely feed it, creating a bird's-eye view of the substations on which 552 facilities and 275 IXPs depend.
To understand which IXPs a grid event takes offline and where their traffic goes, we spread grid events over this map and measure failover with controlled BGP withdrawals on the PEERING testbed.
We find that the dependency of European IXPs is concentrated, as disabling the most-shared 16\% of the substations that feed facilities cuts power to half of all IXP presences (an IXP at one facility), under half of what randomized maps require.
A grid event also reaches far beyond the substations it strikes, since losing the 5\% most central substations cuts power to 0.1\% of IXP presences directly but to 9--52\% once the grid splits into disconnected islands.
Because traffic fails over mostly to nearby sites, which often share the failing grid, a cascade that takes down four metropolitan areas at once sends up to 40\% of the monitored addresses of one area to another of the four.
Building on these findings, we design \name, which strengthens IXPs against grid events by selecting the substations to harden or the facilities whose backup to extend.
Against an attacker who disables ten substations, hardening the ten that \name selects cuts the IXP presence offline by 72\%, 1.4$\times$ the 52\% of the best standard ranking.
\end{abstract}
\maketitle

\section{Introduction}
\label{section:introduction}

On 28 April 2025, the Iberian Peninsula lost its entire power supply~\cite{entsoe2026iberian}.
An overvoltage in southern Spain tripped generation, and the loss of generation cascaded through the grid until the peninsula separated from the rest of the European grid~\cite{entsoe2026iberian}.
Data centers in Spain and Portugal switched to backup power and stayed in operation~\cite{swinhoe2025iberian}, yet traffic at four major Iberian Internet Exchange Points (IXPs) still dropped by 33--70\%, and the disruption reached networks in Morocco and Angola through transit dependencies~\cite{hohlfeld2026lights}.
The grid returned after 12~hours in Portugal and 16~hours in Spain~\cite{entsoe2026iberian}, longer than the 12~hours of fuel that the industry standard requires a data center to store~\cite{heslin2014fuel}.
Without refueling, a Spanish facility with only the required fuel would therefore have gone offline 4~hours before the grid returned.

This chain from the grid to the Internet starts with a fault at a few substations, which can spread through the transmission grid until whole regions lose power.
Each colocation facility\footnote{A \emph{facility} is a building that hosts networking equipment and draws power from the grid. An IXP is a layer-two switching fabric on which networks exchange traffic. It may span several facilities, and a facility may host several IXPs. An \emph{IXP presence} is one IXP at one facility.} in such a region draws its power from nearby high-voltage substations, its \emph{feeders}, and falls back on batteries and generators when they fail~\cite{avelar2016battery}.
Because an IXP operates from one or several facilities, and one facility often hosts several IXPs~\cite{ager2012anatomy}, the loss of one building can affect several exchanges at once.
As the facilities of an IXP run out of backup power, its members are disconnected, and BGP moves their traffic to other \emph{sites} where the same networks interconnect.
That traffic survives only if those sites do not depend on the same failing grid.

The interdependence between the power grid and Internet infrastructure remains unmeasured, because the data and the communities that study the two domains sit apart.
Transmission operators publish their grids in different forms or not at all, the substation that feeds a facility is a private arrangement between the facility and its grid operator, and Internet facilities appear only in operator-maintained registries such as PeeringDB~\cite{peeringdb}.
Power-systems research models cascades inside the grid~\cite{buldyrev_2010,motter2002cascade}, while Internet measurement detects IXP outages at the network layer after they happen, without the grid event behind them~\cite{giotsas2017detecting,milolidakis2019detecting}.
The closest work maps power and Internet infrastructure at the scale of whole synchronous grids~\cite{jyothi2023characterizing}, a scale at which an event inside one synchronous area, such as the Iberian blackout, is not visible.

\emph{The open question is not whether the Internet depends on the grid, but which IXPs a grid event takes down, after how long, and where their traffic goes.}
We answer it with \tool, to our knowledge the first study that maps the interdependence of the power grid and Internet infrastructure from the substation to the IXP and to its traffic.
We first map Europe's high-voltage substations from OpenStreetMap~\cite{openstreetmap}, link each IXP facility to its nearest substations, and check both against the registries of two transmission operators in the region~\cite{tennet,rte}.
To observe where traffic goes, we then emulate the loss of IXP sites in the real Internet with BGP withdrawals on the PEERING testbed~\cite{schlinker19peering}.
Finally, we spread grid events over the measured map and compare their duration with facility backup.
Building on these findings, we design \name, which selects the substations or facilities to harden so that grid events take the least IXP presence offline.

We organize the study around four questions.
\rqbox{1} \emph{Which substations power Europe's IXPs, and how concentrated is this dependence?}
\rqbox{2} \emph{When the IXP facilities of a metropolitan area lose power, where does their traffic go, and does it reach sites on the same grid?}
\rqbox{3} \emph{How much IXP presence does a grid event cut off as it spreads, and when does the outage outlast facility backup?}
\rqbox{4} \emph{Which substations or facilities should an operator harden to prevent the most IXP outages, and how much does that choice depend on the threat?}
Our contributions are fourfold.
\begin{itemize}[leftmargin=0pt,itemindent=0pt,labelwidth=0pt,labelsep=0pt,topsep=.5em,partopsep=0pt,parsep=0pt,itemsep=0pt]
\item[] \rqbox{1}\ \textbf{A substation-level map of the grid-Internet dependency in Europe.}
We extract 20{,}085 high-voltage substations in the scope\footnote{Albania, Austria, Belgium, Bosnia and Herzegovina, Bulgaria, Croatia, Czechia, Denmark, Estonia, Finland, France, Germany, Greece, Hungary, Italy, Latvia, Liechtenstein, Lithuania, Luxembourg, Moldova, Montenegro, the Netherlands, North Macedonia, Norway, Poland, Portugal, Romania, Serbia, Slovakia, Slovenia, Spain, Sweden, Switzerland, and Ukraine, i.e., the Continental European and Nordic synchronous areas. Montenegro has no PeeringDB facility and no facility-level result.} from OpenStreetMap, connect them along the transmission lines, and link each of 552 PeeringDB facilities, which host 275 IXPs, to its three nearest substations as candidate feeders.
To our knowledge, this is the first map of this dependency at substation granularity across a continent.
Against the ground-truth registries of the Dutch and French transmission operators, the map recovers 98.1\% and 96.7\% of the substations at a precision of 86.2\% and 98.5\%, and the candidate feeders contain the feeder that Dutch network operators publish~\cite{capaciteitskaart} for 58 of 61 facilities (\S\ref{section:methodology}).
We measure how concentrated the dependency is by the number of substations that must fail before half of European IXP presence loses power.
When the most-shared substations fail first, 16.2\% of those that feed a facility suffice, whereas maps with the IXP presences shuffled or the candidates rewired need 2.1--2.4$\times$ as many (\S\ref{sec:results:bipartite}).
\item[] \rqbox{2}\ \textbf{Controlled measurements of IXP-site failover.}
We withdraw prefixes from PEERING sites one at a time and in cascades, and observe for 404{,}957 addresses in 61{,}061 ASes where their traffic is rerouted.
Most traffic moves to the nearest sites, which in Europe often share the grid of the failed site.
Some traffic, however, moves thousands of kilometers, because regional transit providers deliver it at distant sites (\S\ref{section:routing_results}).
\item[] \rqbox{3}\ \textbf{Analysis of how far grid events reach.}
We spread grid events over the measured map.
Because whether an island keeps power after a split depends on its balance of generation and load at that moment, we bound island survival with two rules, one matching the 2006 split and one the 2003 and 2025 splits.
Propagation causes most of the outage, as losing the 5\% most central substations cuts power to 0.1\% of IXP presence directly but to 8.8--51.7\% once the grid splits, between the two rules.
With 12~hours of fuel and no refueling, only events longer than 12~hours take IXP presence offline, and the 2003 Italian and 2025 Iberian blackouts both lasted longer.
When a grid cascade takes down four metropolitan areas, 7--40\% of the monitored addresses of each fail over on PEERING to another of the four before reaching a working site (\S\ref{section:grid_results}, \S\ref{sec:shield:routing}).
\item[] \rqbox{4}\ \textbf{\name, hardening built on the measured dependency.}
An operator can harden a limited number of substations, for example with pre-positioned mobile generators~\cite{lei2018mobile}, or extend the backup of a limited number of facilities.
\name chooses them against grid cascades and against an adversary who disables a few substations, as in documented physical attacks (Table~\ref{tab:bg:taxonomy}, Appendix~\ref{app:taxonomy}), following the defender-attacker formulation of infrastructure protection~\cite{brown2006defending}.
Against an adversary who disables ten substations, hardening the ten that \name selects lowers the IXP presence offline from 167 to 47, whereas the best standard ranking leaves 80, 1.7$\times$ as much.
The best set also changes with the threat that the operator plans for (\S\ref{section:shield}).
\end{itemize}

To support reproducibility, we will release our dataset and analysis code upon acceptance.

\section{Background and Related Work}
\label{section:background}

\subsection{Background}

\label{sec:bg:grid}
A transmission system operator (TSO) runs the transmission grid of a country or region, and distribution operators deliver power from it to customers~\cite{eu2019directive944}.
Extra-high-voltage (EHV) lines at $\geq$220~kV carry power over long distances, while high-voltage (HV) lines at 110--150~kV distribute it within a region~\cite{KLAUZER2024104042}.
We call the grid at $\geq$110~kV the \emph{transmission grid}.
The lines terminate at substations, where switchgear connects them and transformers may join one voltage level to the next~\cite{iec60050_605}.
OpenStreetMap (OSM) maps power infrastructure at every voltage, from 380~kV switching yards to distribution kiosks~\cite{osmwiki_substation}.
However, the operator registries against which its coverage can be checked list only transmission substations (\S\ref{section:validation}), and the distribution line that connects a building is a private record of its grid operator.
We therefore link each facility to the transmission substations that likely feed it, the finest grid element we can validate across the scope.

\label{sec:bg:taxonomy}
Because all generators in a synchronous area run at one common frequency~\cite{eu2016regulation631}, a loss of generation or load anywhere in the area shifts the frequency everywhere in it.
The scope spans the Continental European Synchronous Area (CESA) and the Nordic area, which are joined only by direct-current links.
Within an area, a local fault can split the grid into \emph{islands}, parts that lose their connection to the rest of the grid and must balance their own generation and load.
On 4 November 2006, the planned switch-off of a 380~kV double-circuit line in northern Germany overloaded neighboring lines, whose cascading trips half an hour later split the continental grid into three islands within 15~s~\cite{ucte2007split}.
Italy in September 2003~\cite{ucte2004italy} and the Iberian Peninsula in April 2025~\cite{entsoe2026iberian} each separated from the rest of the continental grid, after which the grid of each region collapsed.
Official investigations document eight grid events, from lightning strikes to physical attacks on substations, that lasted from 15~min to several days (Table~\ref{tab:bg:taxonomy}, Appendix~\ref{app:taxonomy}).

\label{sec:bg:backup}
\label{sec:bg:terminology}
A facility rides through a grid failure in two stages, as batteries carry the load until generators start and the generators then run on fuel stored on site~\cite{avelar2016battery} (Table~\ref{tab:bg:backup}, Appendix~\ref{app:backup}).
The Uptime Institute Tier Standard requires fuel for 12~hours at design load at every Tier~\cite{heslin2014fuel,uptime_tier_topology}.
Four of the eight grid events, in Italy, Texas, Ukraine, and Iberia, left part of the affected customers without power for longer than these 12~hours.
Because 255 of the 552 PeeringDB facilities in the scope host more than one IXP, and together hold 79\% of all IXP presences (\S\ref{sec:method:join}), two IXPs that look independent in routing can share a building, its power supply, and its backup.

\subsection{Related Work}
\label{section:related}

\smallskip\noindent\textbf{Mapping Internet and grid infrastructure.}
Prior work maps the long-haul fiber of the US Internet from public records~\cite{durairajan2015intertubes}, assesses the exposure of Internet infrastructure to sea-level rise~\cite{durairajan2018lights}, earthquakes~\cite{mayer2021resilience,9829785}, weather~\cite{padmanabhan2019residential}, and solar superstorms~\cite{jyothi2021solar}, and asks whether LEO satellites can make it more resilient to such hazards~\cite{stevens2024can}.
The closest work~\cite{jyothi2023characterizing} maps power and Internet infrastructure at the granularity of wide-area synchronous grids, finds nearly 65\% of public Internet infrastructure components in fewer than ten such grids, and builds NetWattZap, a tool that suggests deployments that avoid shared grids.
\tool instead maps the dependency at substation granularity, which resolves events inside one synchronous area, such as the Iberian blackout.
Among open grid maps, PyPSA-Eur, an optimization model of the European transmission system, took its grid from the schematic ENTSO-E map~\cite{Horsch_2018} until 2024, when it switched to an extraction from OSM that covers AC lines from 220 to 750~kV~\cite{xiong2025pypsaeur}.
Extractions that reach 110~kV, such as osmTGmod for Germany~\cite{osmTGmod} and the model of Klauzer et al.\ for Austria~\cite{KLAUZER2024104042}, each cover a single country, and GridKit provides general tooling to extract a grid topology from OSM~\cite{gridkit2017}.
These maps serve load-flow, energy-system, and security studies of the grid itself and leave out the Internet facilities that it supplies.
In contrast, \tool extracts every substation at $\geq$110~kV across the scope, which matters because the nearest substation of 74.1\% of facilities operates at 110--220~kV (\S\ref{sec:method:glance}), the tier that a 220~kV floor omits.
It applies the same thresholds in every country and validates its substations and candidate feeders against operator records (\S\ref{section:validation}).

\smallskip\noindent\textbf{Measuring outages and rerouting.}
Prior work detects outages at IXPs and colocation facilities after the fact from BGP data~\cite{giotsas2017detecting,milolidakis2019detecting} and infers which ASes interconnect within a metropolitan area~\cite{salamatian2024metascritic}.
Measurement studies also show that an IXP outage rarely stays local, as only 3\% of the IXPs with dependent networks have all their members and dependents in their own country~\cite{tashiro2024following}.
Routing experiments compare how the BGP announcement strategies of CDNs and clouds trade availability against control when a site fails~\cite{zhu2022best,zhu2025new}.
These works, however, study the loss of a site without the grid event behind it.
With PEERING, we fail the sites that a grid event would take down and test whether their traffic moves to sites that the event also takes down (\S\ref{section:routing_results}).

\smallskip\noindent\textbf{Cascade models.}
Following the interdependency taxonomy of Rinaldi et al.~\cite{rinaldi_2001}, coupled-network percolation models remove the nodes of one network that a failure in another network leaves without support~\cite{buldyrev_2010,Parshani_2010,gao_2012}.
Similarly, interacting dynamical models couple a power grid with a communication network~\cite{rosato_2008}.
Topological and flow-based studies remove grid elements under targeted orderings~\cite{Albert_2004,motter2002cascade,kinney2005modeling,kocc2013entropy}, including in the European grid~\cite{rosas2007topological}.
\tool runs such a cascade on the measured map at continental scale, bounds island survival with two rules matched to the documented grid splits, and compares the duration of the event with facility backup (\S\ref{section:grid_results}).

\smallskip\noindent\textbf{Infrastructure hardening.}
Brown et al.~\cite{brown2006defending} formulate the protection of critical infrastructure as a defender-attacker game.
Later work adds a defender response~\cite{alderson2011dad}, fortifies facility networks~\cite{church2007fortification}, and detects the critical nodes of a network~\cite{arulselvan2009detecting}.
Applied to the power grid, these formulations protect transmission assets against interdiction and measure the loss in unserved load~\cite{salmeron2009worst,yao2007trilevel}.
In contrast, \name protects the Internet infrastructure that depends on the grid, measures the loss in IXP presence, and can harden substations or extend facility backup (\S\ref{section:shield}).

\section{Mapping the Grid-Internet Dependency}\label{section:methodology}

\begin{figure}[t]
\centering
\includegraphics[width=\linewidth]{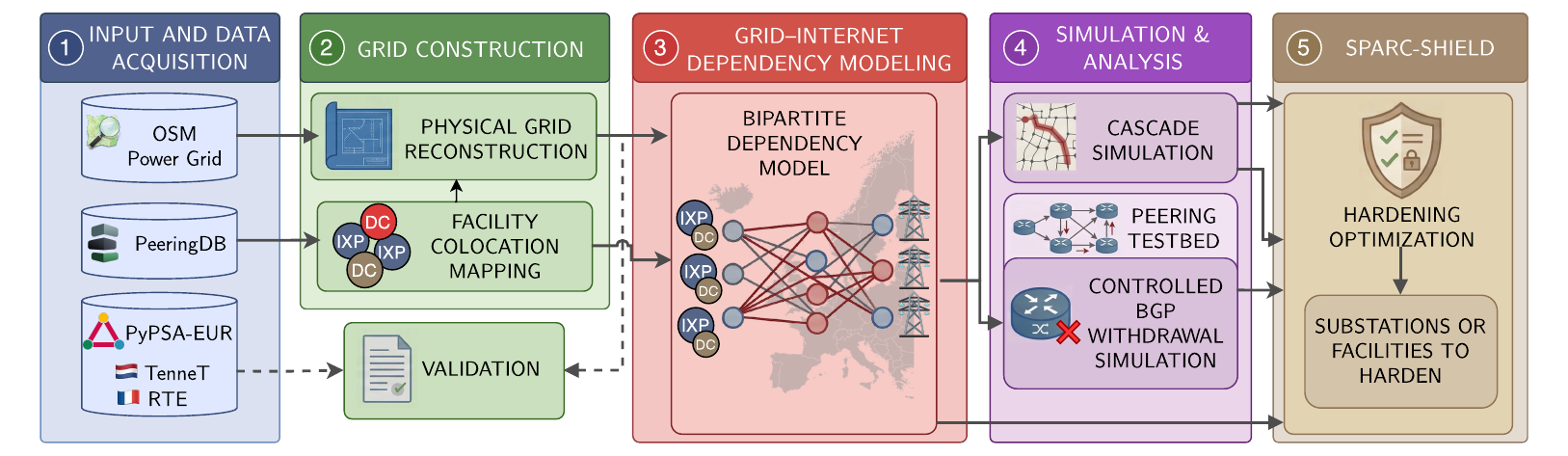}
\caption{The \tool pipeline. \circleone~OpenStreetMap and PeeringDB provide the transmission grid and the Internet facilities (\S\ref{sec:method:data}). \circletwo~We reconstruct the high-voltage substations and lines, map each facility to its location (\S\ref{sec:method:grid}), and validate the substations against PyPSA-Eur and the registries of TenneT and RTE (\S\ref{section:validation}). \circlethree~A bipartite dependency graph links each facility to its candidate feeder substations (\S\ref{sec:method:join}), which we validate against the feeders published in the Netherlands (\S\ref{section:validation}). \circlefour~Controlled BGP withdrawals on the PEERING testbed measure where traffic goes when a site fails (\S\ref{section:routing_results}), and a cascade simulation on the transmission graph determines which substations a grid event de-energizes (\S\ref{section:grid_results}). \circlefive~\name takes the dependency graph, the cascade scenarios, the event duration with the backup of each facility, and the measured rerouting, and selects the substations or facilities to harden (\S\ref{section:shield}).}
\label{fig:methodology}
\vspace{-8pt}
\end{figure}

Figure~\ref{fig:methodology} shows the \tool pipeline.
Because the substation that feeds a facility is a private arrangement with its grid operator, we assemble this dependency from open data in \circleone--\circlethree.
We extract the high-voltage substations and lines of Europe from OpenStreetMap, locate each facility from PeeringDB, and link each facility to its three nearest substations as candidate feeders.
We validate the resulting map against the registries of two transmission operators and the feeder map of the Dutch network operators (\S\ref{section:validation}).
In \circlefour, we measure on the PEERING testbed where the traffic of a failed site goes (\S\ref{section:routing_results}) and follow grid events over this map to the facilities they cut off (\S\ref{section:grid_results}).
In \circlefive, \name selects the substations and facilities to harden (\S\ref{section:shield}).

\subsection{\texorpdfstring{\circleone~}{1.~}Data Sources}\label{sec:method:data}
\label{section:datasets}

\tool takes the transmission grid from OpenStreetMap (OSM)~\cite{openstreetmap}, which volunteers build from aerial imagery, GPS traces, and field surveys.
Beyond roads and buildings, its contributors map power infrastructure, tagging substations with their type and voltage~\cite{osmwiki_substation} and lines with their voltage.
To our knowledge, OSM is the only open source that maps the grid down to 110~kV across the continent with geographic coordinates.
Open grid models of Europe build on it~\cite{gridkit2017,osmTGmod,xiong2025pypsaeur}.

\tool takes the Internet facilities from PeeringDB~\cite{peeringdb}, the registry in which network operators record their interconnection infrastructure.
For each facility, PeeringDB records the location and the IXPs that operate from it, the two attributes we need to link the Internet to the grid.
Lodhi et al.\ compared PeeringDB with BGP data and found that its membership represents transit, content, and access networks across regions and is generally up to date, although small IXPs are under-represented~\cite{lodhi2014using}.
Prior measurement studies use it as the ground truth for facilities~\cite{giotsas2017detecting,milolidakis2019detecting,salamatian2024metascritic}.

\subsection{\texorpdfstring{\circletwo~}{2.~}Power Grid}\label{sec:method:grid}

Because most open grid models extract only the grid at $\geq$220~kV~\cite{scigrid_01,Horsch_2018,xiong2025pypsaeur}, they miss the 110--220~kV tier, the nearest high-voltage source for most facilities (\S\ref{sec:method:glance}).
We therefore build the transmission grid from every substation and line at $\geq$110~kV that OSM maps in Europe.
OSM, however, does not separate transmission from distribution substations, and its lines rarely end exactly at a substation.

\begin{figure}[t]
\centering
\includegraphics[width=\linewidth]{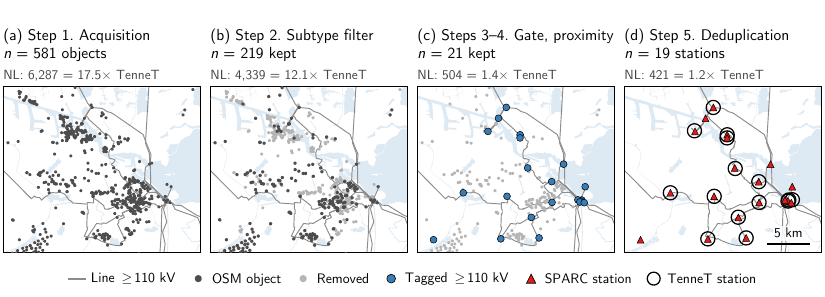}
\caption{The five extraction steps over Amsterdam, with the objects kept in the area ($n$) and, in grey, those kept in the whole Netherlands as a multiple of the 360 stations that TenneT publishes. Each step removes distribution objects, and after Step~5 every TenneT station in the area lies within 364~m of a \tool station.}
\label{fig:extraction_steps}
\vspace{-8pt}
\end{figure}

OSM uses one tag, \code{power=substation}, for everything from a street-level distribution kiosk to a 380~kV switching yard.
We separate the transmission substations from the rest in five steps, which Figure~\ref{fig:extraction_steps} shows on Amsterdam.
Step~1 reduces every OSM object to one point and converts its voltage tags to kilovolts.
Step~2 then removes the subtypes that the OSM tagging convention defines as non-transmission, such as distribution kiosks and railway traction supplies~\cite{osmwiki_substation}, as published OSM grid extractions also do~\cite{gridkit2017,osmTGmod,scigrid_01}.
Step~3 keeps every object tagged at $\geq$110~kV, passes objects with a name but no voltage tag to Step~4, and discards the rest.
Because many real transmission substations carry a name but no voltage tag, Step~4 keeps such an object only if a high-voltage line passes nearby.
The line must lie within 1{,}000~m at $\geq$220~kV and within 500~m at 110--220~kV, a tighter radius because distribution yards often sit next to these lines.
Finally, OSM often maps one substation as several objects, such as busbars and transformer bays.
Step~5 therefore merges objects within 300~m of each other into one substation with DBSCAN and assigns it the highest voltage among them.
The five steps yield 20{,}085 substations in the scope.
Appendix~\ref{app:extraction_steps} gives every rule and the objects each step removes.

OSM lines rarely end exactly at a substation, and joining them at their endpoints would leave the grid in fragments.
We instead build the \emph{transmission graph} by drawing a 2.0~km buffer around each line and connecting the consecutive substations inside it.
GridKit connects a substation to a line that passes over it in the same way~\cite{gridkit2017,wiegmans2016topology}.
Our buffer is wider because we represent each substation by one point and not by its footprint.
The graph keeps the substations and lines beyond the scope boundary so that no line is cut there, and it has 37{,}180 edges with a mean degree of 2.99.

\subsection{\texorpdfstring{\circlethree~}{3.~}Facilities and Their Candidate Feeders}\label{sec:method:join}

From PeeringDB, we take every European facility that at least one IXP lists as one of its sites and whose coordinates the registry records, resulting in 609 facilities.
Within the scope, 552 of them remain, which host 275 IXPs and 1{,}394 IXP presences.
A facility draws power from a distribution grid that nearby transmission substations supply~\cite{rinaldi_2001,jyothi2023characterizing}.
Because OSM does not map the distribution grid at European scale, we cannot trace this path directly.
We therefore link each facility to its $k$ nearest transmission substations within 15~km, its \emph{candidate feeders}, with distances measured in the European projection EPSG:3035.
Together, these links form a \emph{bipartite dependency graph} between facilities and substations, with $k$ ranked edges per facility.
Because the nearest substation of two facilities lies beyond 15~km, our graph covers 550 facilities and 1{,}392 IXP presences.

Facilities sit close to the transmission grid, with the median facility 1.22~km from its nearest substation.
Proximity is therefore a plausible proxy for the feeder, which we test against the feeders that the Dutch network operators publish (\S\ref{sec:method:validation}).
We set $k=3$, because a larger $k$ barely raises the share of Dutch facilities whose published feeder is among the candidates (Table~\ref{tab:feeder_validation}).
Varying $k$ from 1 to 5 changes the share of IXP presence that loses power in our failure analysis by at most 4 percentage points (\S\ref{sec:results:sensitivity}).
The accuracy of the map therefore rests on whether the true feeder is among the candidates, not on how many there are.

\label{sec:method:glance}
Linked this way, the 550 facilities draw on 984 distinct candidate substations, and the nearest of them operates at 110--220~kV for 74.1\% of the facilities.
For the facilities that host 56.8\% of IXP presence, all candidates operate at 110--220~kV, a tier that open grid models with a 220~kV floor omit.
\emph{Only 44 of the 984 candidate substations, under 0.3\% of those in the scope, are the nearest high-voltage source for half of all IXP presences} (Appendix~\ref{app:hotspots}).
Facilities also share whole sets of candidate substations, as 205 facilities holding 49\% of IXP presence have exactly the same set as another facility.
We test in \S\ref{sec:results:bipartite} whether these shared substations and candidate sets make the dependency more concentrated than on randomized versions of the map.

\subsection{Validation}
\label{section:validation}
\label{sec:method:validation}

Because the map rests on volunteer data and a proximity rule, we validate it against references that do not derive from OpenStreetMap.
To test whether \tool places the substations where they are, we compare them with the registries of two transmission operators, TenneT in the Netherlands and RTE in France.
Across the whole scope, we compare the number of substations per country with PyPSA-Eur.
To test whether the candidate set of a facility contains its actual feeder, we use the feeder that the Dutch network operators publish for each postcode (Table~\ref{tab:validation_overview}).

We calibrate the extraction thresholds against the three substation references and apply them unchanged in every country.
TenneT released its stations at $\geq$110~kV for use in OSM~\cite{osmwiki_nl_tennet}, and French volunteers check the grid against the open data of RTE~\cite{osmfr2020reseau}.
The two registry tests therefore measure what the extraction keeps rather than how completely OSM maps a country, which the comparison with PyPSA-Eur measures instead.
Substations are dense in the Netherlands, and almost any point there lies within a few kilometers of one.
A test that matches by distance can thus score well there by chance.
We therefore compare each Dutch test with a \emph{chance baseline}, its score when the \tool substations are displaced or the candidates are drawn at random.
This baseline keeps the density of \tool but not its locations.

\begin{table}[t]
\small
\begin{minipage}[b]{0.56\linewidth}
\caption{Validation of \tool against references independent of OSM, as recall / precision (\%) for substations and the share (\%) of facilities whose candidate set contains the feeder.}
\label{tab:validation_overview}
\end{minipage}\hfill
\begin{minipage}[b]{0.41\linewidth}
\caption{Share (\%) of the 61 Dutch facilities whose published feeder is among their $k$ nearest \tool substations.}
\label{tab:feeder_validation}
\end{minipage}

\smallskip
\begin{minipage}[t]{0.56\linewidth}
\centering
\setlength{\tabcolsep}{3pt}
\begin{tabular}{@{}llcc@{}}
\toprule
Reference & Scope & \tool & Chance \\
\midrule
TenneT & NL, 360 subst. & 98.1 / 86.2 & 15.0 / 13.0 \\
RTE & FR, 930 subst. & 96.7 / 98.5 & -- \\
PyPSA-Eur & Scope, $\geq$220~kV & Count $r$=0.983 & -- \\
\midrule
Capacity map & NL, 61 facilities & 95.1 & 25.7 \\
\bottomrule
\end{tabular}
\end{minipage}\hfill
\begin{minipage}[t]{0.41\linewidth}
\centering
\setlength{\tabcolsep}{3pt}
\renewcommand{\arraystretch}{1.25}
\begin{tabular}{@{}lrrrrrr@{}}
\toprule
Region & $n$ & $k=1$ & $\leq$2 & $\leq$3 & $\leq$4 & $\leq$5 \\
\midrule
Amsterdam  & 20 & 55.0 & 65.0 & 85.0 & 90.0 & 100 \\
Rest of NL & 41 & 68.3 & 90.2 & 100 & 100 & 100 \\
\midrule
All        & 61 & 63.9 & 82.0 & 95.1 & 96.7 & 100 \\
\bottomrule
\end{tabular}
\end{minipage}
\vspace{-8pt}
\end{table}

In the Netherlands, \tool recovers 353 of the 360 substations at $\geq$110~kV in the TenneT registry (recall 98.1\%).
A registered substation counts as recovered when a \tool substation lies within 2~km.
Six of its seven misses are offshore platforms, and a matched \tool substation lies a median 12~m from its registered location, inside the same yard.
Of the 421 Dutch \tool substations, 86.2\% have a TenneT substation within 2~km.
Most of the others are points where a line passes from cable to overhead, which OSM tags as substations but TenneT does not register (Appendix~\ref{app:validation_details}).
Shifting every \tool substation by 5~km lowers recall to the chance baseline of 15.0\% and precision to 13.0\%.
The match thus reflects where \tool places the substations rather than their density.
Across the five extraction steps (Figure~\ref{fig:extraction_steps}), recall stays at 98.1--98.3\% while precision rises from 31.5\% for raw OSM objects to 86.2\%.
The steps thus remove distribution equipment but keep transmission substations.
Appendix~\ref{app:validation_details} repeats the test at four match radii and against random points.

In France, the test matches codes rather than distances, which excludes chance agreement.
RTE publishes a code for each of its 1{,}066 substations at $\geq$150~kV~\cite{rte}, and OSM carries the same code in the tag \code{ref:FR:RTE}~\cite{osmwiki_refrte}.
Of these substations, 98 are traction, generation, distribution, or industrial sites or lie below the voltage threshold, which the extraction removes by design.
Another 38 carry a code absent from OSM (Appendix~\ref{app:validation_details}).
\tool recovers 899 of the remaining 930 (recall 96.7\%), and 899 of the 913 \tool substations with a code match an active RTE substation (precision 98.5\%).

For the comparison of counts across the scope, we use the release of PyPSA-Eur built from the ENTSO-E grid map~\cite{Horsch_2018} and count substations at $\geq$220~kV, the voltage floor of this model.
We use this release because a reference built from OSM would share the gaps and errors of our own source and could not reveal them.
We compare counts and not positions, because this release takes the coordinates of its substations from the GridKit extract of the ENTSO-E interactive map~\cite{Horsch_2018}, whose documentation states that the map favors topological clarity over geographic accuracy~\cite{wiegmans2016entsoe}.
In the Netherlands, where TenneT publishes the coordinates of its substations at $\geq$220~kV, the nearest PyPSA-Eur substation lies a median of 2.1~km from a TenneT substation, against 79~m for the nearest \tool substation.
The counts per country correlate at $r$=0.983, which shows that \tool finds the most substations in the countries where PyPSA-Eur does.
\tool exceeds the PyPSA-Eur count in each of the five countries with the most facilities, whereas in Ukraine, where coverage is thinnest, it finds only 52\% of that count.
In the Netherlands, \tool extracts 57 substations at $\geq$220~kV against the 55 that TenneT registers, whereas PyPSA-Eur lists 33.
Here, the higher \tool count reflects a gap in the model rather than in OSM.
Appendix~\ref{app:validation_details} gives the totals, the rank correlation, and the countries where \tool exceeds the PyPSA-Eur count.

To our knowledge, the only public record of which substation feeds a data center is the capacity map of the Dutch network operators~\cite{capaciteitskaart}.
We can look up the feeders of 61 of the 70 Dutch facilities in this map (Appendix~\ref{app:validation_details} lists the other nine).
Table~\ref{tab:feeder_validation} shows that the nearest substation alone is the feeder for 63.9\% of the 61 facilities and that the share grows with each added candidate.
At $k=3$, the value we use (\S\ref{sec:method:join}), the candidate set contains the feeder for 95.1\% of the 61 facilities (95\% Wilson interval 86.5--98.3\%).
The three facilities whose feeder the candidates miss lie within 20~km of the center of Amsterdam, where the feeder is the fourth or fifth nearest substation.
Three substations drawn at random within 15~km of each facility contain the feeder for only 25.7\% of the facilities, the chance baseline of this test.
Because the map names one feeder per postcode, the test checks only the primary feeder and not the backups.
Together, the tests show that \tool places substations where both operators register them and that its counts per country track PyPSA-Eur across the scope.
The candidate sets of \tool also contain the published feeder for 58 of the 61 Dutch facilities whose feeder the capacity map records.

\begin{tcolorbox}[title={Takeaway \rqbox{1}}]
Most of Europe's IXP facilities are nearest to substations at 110--220~kV, a tier that open grid models with a 220~kV floor omit, and the facilities that host over half of European IXP presence have all their candidate feeders in this tier.
The dependency is also concentrated, as under 0.3\% of the mapped substations are the nearest high-voltage source for half of IXP presence.
\end{tcolorbox}

\section{Characterizing Internet Failover under Site Outages}\label{section:routing_results}

The map of \S\ref{section:methodology} links each IXP facility to the substations that may feed it, but it does not show where the traffic at these facilities goes when they lose power.
When a grid event lasts longer than the backup of these facilities, the networks that interconnect at them lose their BGP sessions, and BGP moves their traffic to other sites where the same networks interconnect.
These networks remain reachable only if the other sites draw power from a part of the grid that the same event does not take down.
The sites that receive the traffic depend on the routing decisions of each network, which become visible only when a site fails.
Because we cannot fail a production site, we emulate site outages with BGP withdrawals on the PEERING testbed~\cite{schlinker19peering}.
We measure where the addresses served by a failed site move and on which power grid their new sites lie, both when one site fails and when its nearby sites fail with it.

\subsection{\texorpdfstring{\circlefour~}{4.~}Emulating Site Outages on PEERING}\label{sec:method:peering}

\begin{wraptable}{r}{0.42\linewidth}
\vspace{-33pt}
\centering
\small
\setlength{\tabcolsep}{3pt}
\renewcommand{\arraystretch}{1.0}
\begin{tabular}{@{}llrrr@{}}
\toprule
Site & Area & Fac. & IXP pres. & Subst. \\
\midrule
Amsterdam & CESA & 19 & 156 & 7 \\
Frankfurt & CESA & 30 & 159 & 11 \\
Madrid & CESA & 9 & 29 & 4 \\
Paris & CESA & 16 & 37 & 9 \\
Warsaw & CESA & 10 & 52 & 7 \\
Stockholm & Nordic & 8 & 51 & 5 \\
\bottomrule
\end{tabular}
\caption{PEERING sites in the scope, with their synchronous area, the PeeringDB facilities within 10~km, the IXP presences at these facilities, and their distinct nearest substations. London and Manchester lie outside the scope.}
\label{tab:peering.site-grid}
\vspace{-10pt}
\end{wraptable}
PEERING lets researchers announce their own IP prefixes with BGP and exchange routes and traffic with hundreds of neighboring networks.
Our experiments use the 32 PEERING sites in data centers of the Vultr cloud, which span all six continents.
Six of these sites lie in metropolitan areas of the scope, where the facilities within 10~km of each site hold 484 IXP presences, a third of those in the map (Table~\ref{tab:peering.site-grid}).
Because we do not know through which facilities of its metropolitan area a Vultr data center connects, we let each PEERING site stand for its whole metropolitan area.

We run two experiments from an anycast prefix that every site announces to all its providers.
Withdrawing this prefix from a site emulates the outage of its metropolitan area once the facilities of the area have exhausted their backup.
The traffic toward the prefix then moves to the sites that still announce it, which we call \emph{alternate sites}.
In \emph{isolated outages}, we withdraw the prefix from one site at a time, as when a metropolitan area loses power while the rest of the grid holds.
In \emph{cascading outages}, we withdraw the prefix first from one site, the \emph{epicenter}, and then from further sites one at a time in order of increasing distance.
Ordering by distance approximates a grid event that spreads across borders, as the 2006 split of the continental grid did~\cite{ucte2007split}, although a real event follows the transmission lines and the islands that they form.
From Amsterdam, for instance, the order fails London second, although Great Britain is a separate synchronous area that only direct-current links join to CESA.
We start cascades from Amsterdam, with the second-most IXP presences among the sites in the scope (Table~\ref{tab:peering.site-grid}), and from Madrid, inside the region of the Iberian blackout.
Outside Europe, we start them from New Jersey, with three nearby alternate sites, and Sydney, with one.

Before each step of either experiment, we withdraw the prefix from every site for 15~minutes so that every step starts from the same routing state, and we then keep the announcement configuration of the step active for 90~minutes (Appendix~\ref{app:peering}).
In every step, we ping 404{,}957 responsive addresses in 61{,}061 ASes, each in a distinct globally announced prefix.
These \emph{monitored addresses} belong to other networks, and their number is bounded by the probing rate that PEERING allows.
We record the site at which each reply arrives, and the monitored addresses whose replies arrive at a site form its \emph{catchment}.
We start measuring 15~minutes after each announcement, once routes have converged~\cite{garcia19convergence,bassett12lifeguard}.
To see the paths behind each catchment, we keep the first BGP route toward our prefixes that stays valid for 5~minutes~\cite{feldman04bgp} at the RIPE RIS~\cite{riperis} and RouteViews~\cite{routeviews} collectors.
We complement these routes with traceroutes from 4{,}259 RIPE Atlas probes~\cite{ripeatlas} every 20~minutes, and we map addresses to ASes with \textsf{bdrmapit}~\cite{marder18bdrmapit} and to countries with IPinfo~\cite{ipinfo}.

From these catchments, we derive the \emph{reroute distance} of a monitored address, the distance from the failed site to the site whose catchment the address joins after the outage.
The \emph{latency change} of a monitored address is its latency to the new site minus its latency to the failed site, each averaged over the monitored addresses of its $\langle$AS, country$\rangle$ pair.
In cascading outages, we measure both metrics against the epicenter and call the latency change there the \emph{latency inflation}.
We measure the latency from each site to every monitored address by announcing a prefix from that site alone, pinging each address from our router, and subtracting the latency between the router and the site (Appendix~\ref{app:peering}).
Withdrawing the prefix removes a whole metropolitan area at once, whereas its facilities lose power at different times as their backup runs out (\S\ref{section:limitations}).
Because the traffic on most inter-AS links is invisible to public measurements~\cite{salamatian2024metascritic,koch2021towards}, we count monitored addresses rather than traffic volume.
Because PEERING connects to other networks much as a cloud provider does~\cite{schlinker19peering}, our measurements show where the addresses served by a cloud network move when one of its sites fails, and how much latency the move adds.

\subsection{Alternate Sites and Their Power Grid}\label{sec:peering:isolated}
\label{sec:peering:grid}
\begin{figure}[t]
\centering
\includegraphics[width=\linewidth]{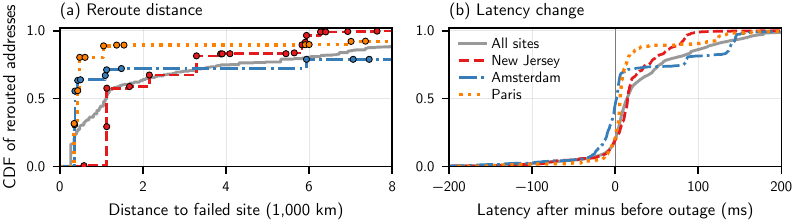}
\caption{Isolated outages, one PEERING site withdrawn at a time, for all failed sites together and for three of them. (a) Distance from the failed site to the alternate site that receives each rerouted monitored address, with one circle per alternate site. (b) Latency change, averaged per $\langle$AS, country$\rangle$ pair. Because most addresses move to a nearby site, over half of all addresses change latency by at most +20~ms.}
\label{fig:peering.isolated}
\label{fig:peering.isolated.distance}
\label{fig:peering.isolated.latency}
\vspace{-8pt}
\end{figure}

When a single site fails, most of its catchment moves to its four nearest alternate sites.
Over the 29 isolated outages, these four sites receive a median of 67\% of the catchment, and at least half in 23 outages (Table~\ref{tab:app.peering.results}, Appendix~\ref{app:peering.results}).
Figure~\ref{fig:peering.isolated}a shows these moves as steep first steps in the distribution of the reroute distance.
Grouping the addresses by AS or weighting each AS by its estimated user population changes this distribution little (Appendix~\ref{app:peering.results}).
How the addresses divide among the four sites depends on the routes that their networks choose rather than on distance.
When Warsaw fails, for example, its nearest alternate site, Stockholm, receives 2.0\% of its catchment, whereas Amsterdam, the third nearest, receives 40.3\%.

Some addresses move to another continent although nearby sites remain up, because after the outage their networks reach PEERING through a regional transit provider that connects to PEERING at only a few sites.
When Amsterdam fails, for example, 15\% of its catchment moves to Mexico, including the addresses of SURFnet (AS1103) and SeeWeb (AS12637).
Both networks exchange traffic with PEERING directly in Amsterdam, and after the outage they reach PEERING through EdgeUno (AS7195).
EdgeUno delivers their traffic at the Mexico site, one of the few sites where it provides transit to PEERING.
The destination of these detours thus depends on where PEERING buys transit rather than on distance.

Even before an outage, 60\% of the monitored addresses in Europe reach a site in another country, and an outage of that site moves them even when the grid of their own country keeps power.
Most of these addresses, 51\% of those in Europe, lie in a country without a PEERING site, and the other 9\% are routed abroad by their networks although their country has one.

To see whether failover leaves the grid of the failed site, we label each European site with its synchronous area (Table~\ref{tab:peering.site-grid}), and Figure~\ref{fig:peering.cross-grid} gives the area of the four nearest alternate sites of each.
\begin{figure}[t]
\centering
\includegraphics{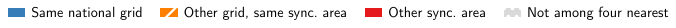}\par
\vspace{3pt}
\begin{minipage}[b]{0.57\linewidth}
\centering
\includegraphics[width=\linewidth]{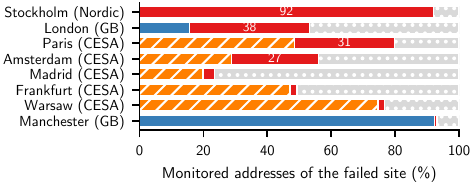}
\end{minipage}\hfill
\begin{minipage}[b]{0.41\linewidth}
\caption{Monitored addresses of each failed European PEERING site, by the grid of the alternate site that receives them, among its four nearest. For every site in CESA, most addresses that reach these sites stay in CESA (hatched), whereas most addresses of Stockholm, whose nearest sites lie in CESA, move to another synchronous area.}
\label{fig:peering.cross-grid}
\end{minipage}
\vspace{-8pt}
\end{figure}
Because PEERING has one site per country in the scope, a site in the scope and each of its alternate sites lie on different national grids.
When one of the five sites in CESA fails, \emph{51--97\% of the addresses that move to its four nearest alternate sites stay in CESA}, and these addresses make up 20--75\% of its catchment.
Although these alternate sites lie in other countries, CESA operates as one synchronous grid, which a single event split into three islands within 15~s in 2006~\cite{ucte2007split}.

When Stockholm fails, in contrast, 92.0\% of its catchment moves to CESA, because three of its four nearest alternate sites, Warsaw, Amsterdam, and Frankfurt, lie in CESA.
When Paris or Amsterdam fails, 27--31\% of its catchment moves to Great Britain, against under 4\% for the other three sites in CESA.
The shares that stay inside CESA are lower bounds, because 8--77\% of the catchment of a failed European site moves beyond its four nearest alternate sites, to sites whose grid Figure~\ref{fig:peering.cross-grid} does not show.
\name uses the share of each European catchment that moves to each other European site to count the IXP presence whose traffic fails over to a site that the same grid event takes down (\S\ref{sec:shield:routing}).

\subsection{Isolated versus Cascading Outages}\label{sec:peering:cascading}

Because most addresses move to a nearby site after an isolated outage, the median latency change over all isolated outages is +14~ms, and 56\% of the monitored addresses change by at most +20~ms (Figure~\ref{fig:peering.isolated}b).
In contrast, the 14\% that move to a distant site gain 100~ms or more.
For another 20\%, the latency falls, because the failed site was not their fastest site.

In cascading outages, the reroute distance and the latency inflation stay near their isolated values while a nearby alternate site remains up, and \emph{rise in a step once the last nearby alternate site fails} (Figure~\ref{fig:peering.cascading}).
\begin{figure}[t]
\centering
\includegraphics[width=\linewidth]{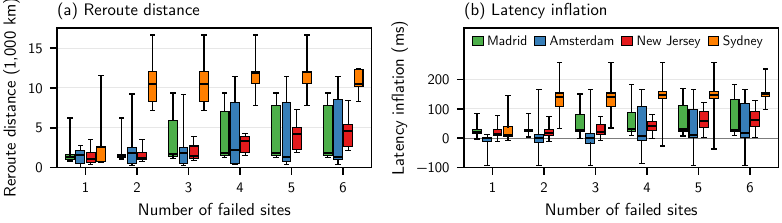}
\caption{Cascading outages for the monitored addresses in the catchment of each epicenter, as further sites fail in order of distance, starting with the epicenter alone (boxes give the 25th--75th percentile, whiskers the 5th and 95th). (a) Reroute distance. (b) Latency inflation. Both stay low while a nearby alternate site remains up and rise in a step once the last one fails, at the second failure for Sydney and the fourth for Amsterdam.}
\label{fig:peering.cascading}
\label{fig:peering.cascading.distance}
\label{fig:peering.cascading.latency}
\vspace{-8pt}
\end{figure}
For Amsterdam, whose catchment moves first to London, Frankfurt, and Paris, the 75th percentile of the reroute distance grows from ${\sim}$2{,}500 to ${\sim}$8{,}200~km at the fourth failure, once all three are down.
At the same failure, the 75th percentile of the latency inflation rises from ${\sim}$15 to ${\sim}$107~ms.
For Madrid, the step comes at the third failure, once Paris and London are down.
Outside Europe, the step comes at the fourth failure for New Jersey, after Toronto, Chicago, and Atlanta, and at the second for Sydney, after Melbourne.
The position of the step thus follows the number of nearby alternate sites, from one for Sydney to three for Amsterdam and New Jersey.

\begin{tcolorbox}[title={Takeaway \rqbox{2}}]
When the IXP facilities of a metropolitan area lose power, a median of two thirds of the monitored addresses that they serve move to the four nearest sites, and in continental Europe more than half of these addresses stay on the same synchronous grid.
The addresses reach another grid when the nearest sites lie in another synchronous area, or when the grid event takes these sites down as well, after which the addresses move to sites thousands of kilometers away.
\end{tcolorbox}

\section{Failure Analysis}
\label{section:grid_results}
\label{section:model}

The map of \S\ref{section:methodology} links each IXP facility to the substations that may feed it, and the PEERING experiments of \S\ref{section:routing_results} show that the traffic of a failed site moves mostly to nearby sites, often on the same synchronous grid.
The map and the experiments, however, do not show which facilities a grid event takes offline, or for how long, and an operator needs both answers to plan backup power.
We answer both questions on the measured map by following a grid event through the steps of Figure~\ref{fig:model}.
To compare an attack with a cascade, we describe both by the set $X$ of substations that the event de-energizes in step \stepA.
For an attack, $X$ holds the substations that it strikes, and for a cascade also every substation of the islands that lose power.
A facility then loses power (\stepB), and it goes offline only when the event outlasts its backup (\stepC).
How far an event spreads depends on the state of the grid when it occurs, and how long a facility holds out depends on the fuel that its operator stores and does not publish.
We therefore bound the spread by the documented splits of the continental grid and the backup by published fuel levels, and report each result as a range.

Because operators do not publish which substation supplies a facility, we call facility $f$ \emph{dark} when every substation in its candidate set $N(f)$, its three nearest transmission substations within 15~km (\S\ref{sec:method:join}), lies in $X$.
Facility $f$ hosts $w_f$ IXP presences, and we measure an event by the share of IXP presence at dark facilities,
\begin{equation}
\Phi(X) = \frac{\sum_{f \in F} w_f \, \mathbf{1}[N(f) \subseteq X]}{\sum_{f \in F} w_f}.
\label{eq:ixp_offline}
\end{equation}
We count IXP presences rather than IXPs because most IXPs operate from several facilities, and an IXP loses only part of its presence when one of them goes dark.
To report whole IXPs, we call an IXP \emph{down} when all its facilities are dark, and \emph{degraded} when some but not all are.
A degraded IXP keeps exchanging traffic at its powered facilities and disconnects only the members whose ports all lie in dark facilities.
When the candidate set contains the substation that actually supplies the facility, every facility that we call dark has truly lost power, and $\Phi$ can therefore understate an outage but not overstate it.
Adding candidates can only lower $\Phi$ (Proposition~\ref{prop:conservative}).

\subsection{Direct Loss of Shared Substations}
\label{sec:model:system}
\label{sec:results:bipartite}

\begin{figure}[t]
    \centering
    \includegraphics[width=0.8\linewidth]{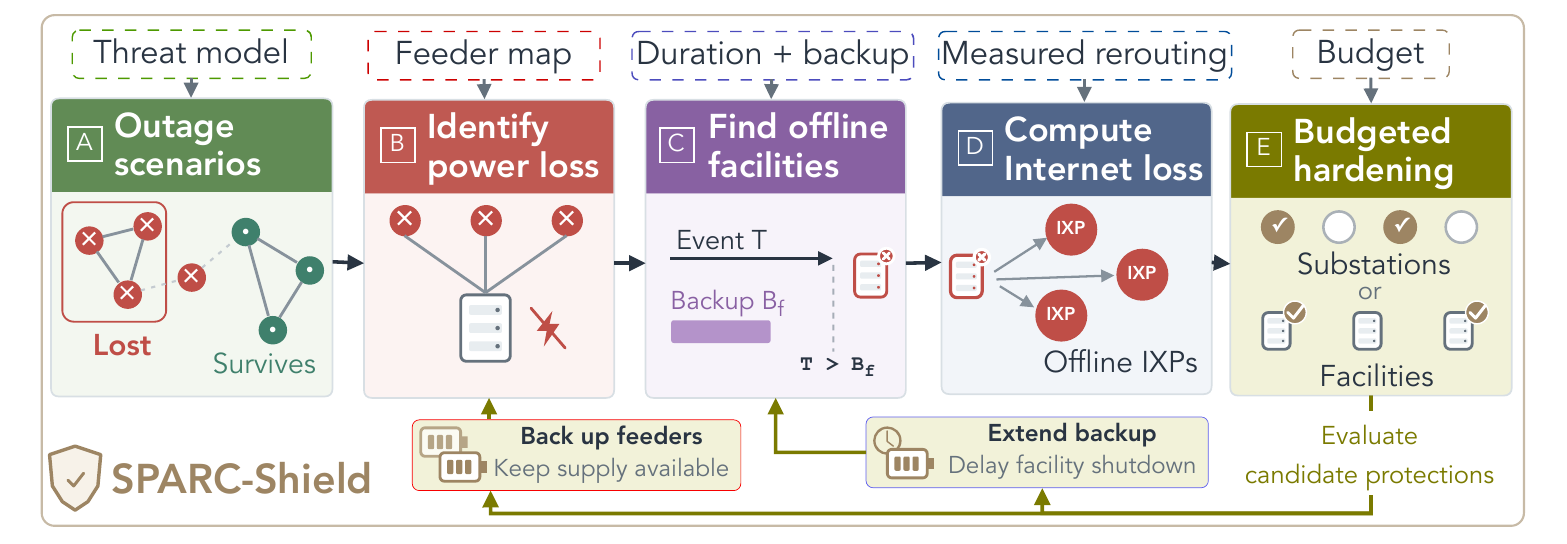}
    \caption{The steps by which a grid event reaches IXPs and their traffic, from the outage scenario (\stepA) through the facilities that lose power (\stepB) and go offline once their backup runs out (\stepC) to the Internet loss under the failover of \S\ref{section:routing_results} (\stepD) and the budgeted hardening (\stepE). The failure analysis runs steps \stepA to \stepC without hardening, and \name (frame) reruns them with hardening applied.}
    \label{fig:model}
    \vspace{-8pt}
    \end{figure}

Documented attacks on the grid disable a few substations directly, as in the rifle attack on the Metcalf substation in 2013 and the strikes on Ukrainian substations since 2022 (Table~\ref{tab:bg:taxonomy}).
How much IXP presence such attacks can cut off depends on how concentrated the dependency of facilities on substations is, which we measure as the number of substations that must fail before half of European IXP presence loses power.
We remove the 984 facility-connected substations one at a time, starting with the \emph{most-shared}, those that appear in the candidate sets of the most facilities, as in the targeted removal of network robustness studies~\cite{albert_2000,holme2002attack,schneider2011mitigation}.
When substations appear in equally many candidate sets, the one whose facilities host the most IXP presence goes first.
In this order, \emph{half of IXP presence loses power after 159 of the 984 substations (16.2\%)}, whereas random removal needs 780 (79.3\%) on average over 1{,}000 orders (Figure~\ref{fig:eu_attack_tolerance}, Appendix~\ref{app:calibrated}).
Breaking ties at random instead of by IXP presence raises the 16.2\% to 21.3--27.9\% over 200 orders, because many substations appear in equally many candidate sets.

Random removal, however, does not measure concentration, because the chance that it removes all three candidates of a facility depends on their number, not on which substations they are.
Random removal therefore cuts power to half of IXP presence after ${\sim}$79\% of the substations on any map with three candidates per facility.
To test whether the IXPs themselves concentrate on shared substations, we repeat most-shared removal on two altered maps, each drawn 1{,}000 times.
The \emph{shuffled map} moves the IXP presences across facilities but keeps which ones share a substation.
The \emph{rewired map} reassigns the candidates while every facility keeps its number of candidates and every substation its number of facilities, which preserves how many facilities share each substation but not which ones.
On the shuffled and rewired maps, most-shared removal needs 33.6\% and 39.3\% of the substations, and each of the 2{,}000 draws needs more than the measured 16.2\%.
The measured map needs fewer substations because its facilities with the most IXPs share their candidates.

Because of this sharing, facilities lose power more slowly than IXP presence, and half of the facilities are dark only after 38.3\% of the substations are removed.
Removing the first 159 substations takes 40 of the 275 IXPs down and degrades 104, because most of these IXPs keep some of their facilities.
Assigning three candidates to each facility could overstate how many substations must fail, because a real facility often draws its power from a single substation.
When we instead assign each facility one candidate, drawn with the rank probabilities measured on the Dutch capacity map (\S\ref{section:validation}), most-shared removal needs 12.5\% of the substations (10.8--14.0\% over 200 draws) and random removal 51.3\%.
Across every variant of the candidate set that we test, most-shared removal needs 2.5--4.9$\times$ fewer substations than random removal (Appendix~\ref{app:sensitivity_feeder}).
Within Amsterdam and Frankfurt, most-shared removal cuts power to half of the metro's IXP presence after 35\% of its candidate substations, against 81\% for random removal, and within Paris after 41\% (Appendix~\ref{app:national}).

\begin{figure}[t]
\centering
\includegraphics[width=\linewidth]{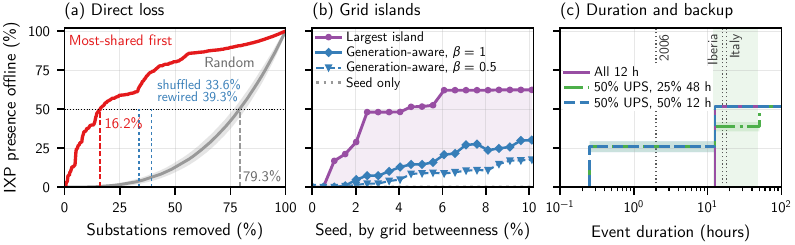}%
{\phantomsubcaption\label{fig:eu_attack_tolerance}}%
{\phantomsubcaption\label{fig:cascade_band_main}}%
{\phantomsubcaption\label{fig:outage_duration_main}}
\caption{IXP presence that a grid event cuts off (a) when it strikes facility-connected substations, (b) when a betweenness seed splits the grid, as a band between the two island rules, and (c) against event duration, with documented events dotted. Most-shared removal cuts power faster than random removal (band, $\pm 1$ s.d.) and than the same removal on the altered maps, islanding cuts off far more than the seed alone, and IXP presence goes offline only once the event outlasts backup, inside the shaded window of \name.}
\label{fig:failure_experiments}
\vspace{-8pt}
\end{figure}

\subsection{Spread Through Grid Islands}
\label{sec:method:cascade}
\label{sec:results:cascade}
\label{sec:model:failure}

Unlike an attack, a cascade reaches the substations that facilities depend on through the grid, without striking them.
When a few transmission substations fail, the grid can split into \emph{islands}, disconnected parts that must each balance their own generation and load, and an island that cannot do so loses power with every substation inside it.
We model islanding on the backbone, the largest connected part of the transmission graph, which stands in for the continental synchronous grid (Appendix~\ref{app:backbone}).
Substations with a high \emph{grid betweenness}, the share of shortest paths in the transmission graph that pass through them, are the most likely to split the grid when they fail.
We therefore seed a cascade with the set $S_0$ of the top $q\%$ of backbone substations by grid betweenness, and we compare this seed with grid-degree and random seeds.
Removing $S_0$ leaves the islands $C_1,\dots,C_m$, in descending size.
Of the 5\% of backbone substations with the highest betweenness, only 33 are facility-connected, and only 5 are among the 159 that most-shared removal takes first (\S\ref{sec:results:bipartite}).
Within Amsterdam, Frankfurt, and Paris, removing the candidate substations of each metro in order of grid betweenness cuts power as slowly as random removal.
A ranking by grid structure thus misses the substations that most facilities share.
Only the 382 facilities whose three candidates all lie in the backbone can lose power in a cascade, which caps the IXP presence that any cascade can take offline at 65.2\% (Proposition~\ref{prop:ceiling}).
Until \S\ref{sec:results:duration}, we count every dark facility as offline.

The documented splits of the continental grid ended in opposite ways, as the three islands of 2006 kept power after shedding load~\cite{ucte2007split}, whereas Italy in 2003 and the Iberian Peninsula in 2025 lost power after they separated~\cite{ucte2004italy,entsoe2026iberian}.
We therefore bound island survival with one rule for each outcome.
The \emph{largest-island rule}, from the two-stage cascade model~\cite{buldyrev_2010}, keeps only the largest island $C_1$ energized, and every other island loses power with the seed,
\begin{equation}
X_{\mathrm{isl}} = S_0 \cup \bigcup_{j \geq 2} C_j .
\label{eq:two_stage}
\end{equation}
The largest-island rule matches 2003 and 2025, because Italy and Iberia each separated as a smaller island and lost power.
The \emph{generation-aware rule} instead keeps island $C_j$ energized when the dispatchable capacity $G_j$ sited in it covers at least a share $\beta$ of its load $L_j$,
\begin{equation}
X_{\mathrm{gen}}(\beta) = S_0 \cup \bigcup_{j \,:\, G_j < \beta L_j} C_j ,
\label{eq:island_gen}
\end{equation}
where $\beta<1$ allows an island to shed load, as the islands of 2006 did.
We take the capacity of each island from the plants of powerplantmatching~\cite{gotzens2019ppm} and its load from the 2019 peak load of each country~\cite{opsd2020timeseries}, spread over the substations of that country.
The generation-aware rule matches 2006, because the three islands of 2006 hold 1.06--1.34$\times$ their load in present-day capacity.
For Italy, however, the generation-aware rule fails because it counts installed capacity, 1.30$\times$ the Italian load today, whereas Italian plants generated only 0.75 of the load when Italy separated in 2003, against 0.95--1.19 in the islands of 2006.
Each rule thus reproduces only one of the two outcomes (Table~\ref{tab:historical_check}).
Both rules ignore power flow, dispatch, and line ratings, and \name evaluates hardening under each of them (\S\ref{section:shield}).

\begin{figure}[t]
\centering
\includegraphics[width=0.97\linewidth]{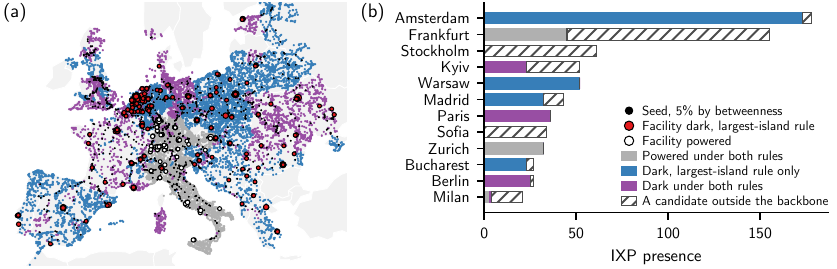}
\caption{(a) The backbone after a 5\% betweenness seed, with the facilities that a cascade can reach and each island colored by its outcome under the two island rules. (b) IXP presence of the twelve largest metros by outcome. Only the grey Alpine island keeps power under both rules, and hatched IXP presence keeps power under either rule because its facilities have a candidate outside the backbone.}
\label{fig:islands_map}
\vspace{-8pt}
\end{figure}

Removing the 5\% of backbone substations with the highest betweenness splits the backbone into 646 islands (Figure~\ref{fig:islands_map}a).
Of these islands, the largest-island rule keeps only the Alpine island of 4{,}629 substations energized, which holds the facilities of Frankfurt, Zurich, and Milan that a cascade can reach (Figure~\ref{fig:islands_map}b).
Every other island loses power, and the facilities whose candidates all lie in these islands host 51.7\% of IXP presence, among them those of Amsterdam, Warsaw, Paris, Berlin, Madrid, and Kyiv.
Under the largest-island rule, 88 of the 275 IXPs go down and 71 are degraded.
These 88 IXPs are more than twice the 40 that most-shared removal takes down at half of IXP presence, because the loss of an island darkens every facility of a metro at once.
Under the generation-aware rule with $\beta=1$, 151 of the 646 islands keep power and 14.6\% of IXP presence goes offline, with 28 IXPs down and 43 degraded.
Paris, Berlin, and Kyiv still go dark, because the capacity sited in their islands does not cover their load.
In Frankfurt, 110 of 155 IXP presences keep power under either rule, because each of their facilities has a 110~kV candidate that the OpenStreetMap lines do not connect to the backbone (hatched in Figure~\ref{fig:islands_map}b).
Across seed sizes (Figure~\ref{fig:cascade_band_main}), half of IXP presence goes offline under the largest-island rule at a 4.6\% betweenness seed, against ${\sim}$8\% for grid-degree seeds and 25\% for random seeds.
At a 5\% seed, 8.8--14.6\% goes offline under the generation-aware rule as $\beta$ rises from 0.5 to 1.
\emph{Together, the two rules place the outage at 8.8--51.7\% of IXP presence, whereas the seed alone takes 0.1\% offline}.

\subsection{Duration Against Facility Backup}
\label{sec:results:sensitivity}
\label{sec:model:duration}
\label{sec:results:duration}

A dark facility goes \emph{offline} only when the duration $T$ of the event exceeds its \emph{backup autonomy} $B_f$, the hours that its batteries and stored fuel keep it running.
Because colocation operators do not publish the backup autonomy of their facilities, we use the three classes of Table~\ref{tab:bg:backup} (Appendix~\ref{app:backup}).
The classes are UPS only (0.25~hours), the 12-hour fuel minimum of the Uptime Tier Standard~\cite{heslin2014fuel}, and 48~hours of stored fuel, the most that one hyperscale operator states for its data centers~\cite{microsoft2022fuel}.

With 12~hours of fuel at every facility, the IXP presence offline under the largest-island rule at a 5\% seed is a step, \emph{with nothing offline until 12~hours and all 51.7\% after it} (Figure~\ref{fig:outage_duration_main}).
The 2006 split, after which the grid was restored within 2~hours, lies left of the step, whereas the 2003 split and the 2025 Iberian blackout, which lasted 18 and 16~hours, lie right of it (Table~\ref{tab:bg:taxonomy}).
Extending autonomy from 12 to 48~hours therefore changes the outcome only for events of 12--48~hours, the window in which \name extends backup (shaded, \S\ref{section:shield}).
When autonomy differs across facilities, the step splits in two.
When half of the facilities hold UPS only, a 2-hour event already takes 25.8\% offline (22.3--29.5\% over 200 assignments, Appendix~\ref{app:duration}).
When a quarter of the facilities hold 48~hours instead of 12~hours, a 16-hour event takes 38.6\% offline rather than 51.7\%.

We test the duration rule on the Iberian blackout, the only documented event in the scope that outlasted 12~hours of fuel and has measurements of IXP traffic~\cite{entsoe2026iberian,hohlfeld2026lights}.
Removing every substation in Spain and Portugal darkens 29 facilities.
With 12~hours of fuel, the 13 IXPs that operate only from Spain go offline, because the Spanish grid returned after 16~hours.
The 6 IXPs that operate only from Portugal sit exactly at the threshold, because the Portuguese grid returned after 12~hours.
Eight operators nonetheless kept their Iberian data centers running~\cite{swinhoe2025iberian}, on more than 12~hours of fuel or on local supply that returned before the full grid.
The 12-hour minimum therefore gives the pessimistic end of the outage.
Traffic at four major Iberian IXPs still fell by 33--70\%~\cite{hohlfeld2026lights}, which shows that exchanged traffic falls even when facilities keep power.
For events shorter than facility backup, our analysis thus understates the loss of exchanged traffic.

At a 5\% betweenness seed, varying island survival or event duration alone changes the IXP presence offline from almost none to 51.7\% (Table~\ref{tab:sensitivity_full}).
When overloaded substations also fail before the islands form, the outage under the largest-island rule rises to 64.7\%.
The other inputs that we test are the number of candidates $k$, how they are chosen, the corridor buffer, the seed type, and the scope of the graph, and all but one change it by at most 6 percentage points.
Only the variant that links each facility to every substation within 15~km, a median of 22 candidates, lowers it to 25.6\%, because more facilities then have a candidate outside the backbone.

\begin{tcolorbox}[title={Takeaway \rqbox{3}}]
A split of the continental grid cuts power to between a tenth and half of European IXP presence, depending on which islands balance their generation and load, whereas the substations at which the split starts supply almost none of it.
The facilities in these islands go offline only when the grid stays down longer than their backup, which at the 12-hour fuel minimum of the industry standard happens in blackouts that last most of a day but not in splits that are restored within a few hours.
\end{tcolorbox}

\section{\name: Hardening Substations and Facilities}\label{section:shield}
\label{section:discussion}
\label{sec:discussion:hardening}

The failure analysis (\S\ref{section:grid_results}) showed two ways in which a grid event takes IXP presence offline.
An attack on a few substations that many facilities share cuts power to their facilities, and a split of the grid de-energizes whole islands.
In both cases, a facility goes offline only when the event outlasts its backup (\S\ref{sec:results:duration}).
Longer backup at selected substations or facilities can therefore prevent part of the outage.
Because an operator cannot extend backup everywhere, the question is \emph{which few substations or facilities to harden, and whether the answer depends on the threat}.
\name answers this question by repeating steps \stepA to \stepC of the failure analysis (Figure~\ref{fig:model}) with hardening applied at steps \stepB and \stepC.
It then selects, under a budget, the substations or facilities whose hardening minimizes the IXP presence offline (\stepD, \stepE).
The substations have to be selected together, because the protection that one of them adds depends on which others are already hardened.
For example, backing up the ten substations that serve the most IXP presence lowers the loss of an optimal attack on ten substations by only 51\% (Figure~\ref{fig:shield_curves}a).
The attack then moves to lower-ranked substations that are not backed up.

\begin{figure}[t]
\centering
\includegraphics[width=\linewidth]{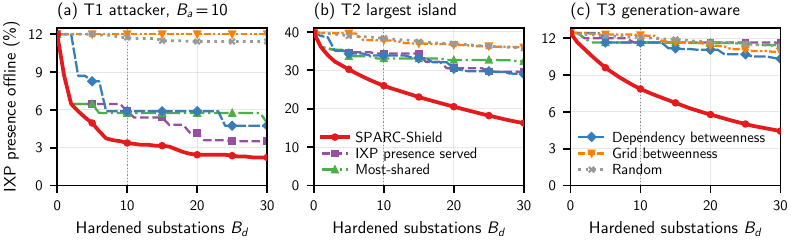}
\caption{IXP presence offline against the number of hardened substations $B_d$ under the three threat models (cascades averaged over four seeds), for a 16-hour event and 12~hours of facility backup. Above two substations, \name has the lowest loss, because the rankings level off once the attack moves (a) or the loss shifts to islands short of generation (c).}
\label{fig:shield_curves}
\vspace*{-1.7em}
\end{figure}

\subsection{Problem Formulation}
\label{sec:shield:formulation}

Because the substations that many facilities share are rarely those at which a cascade starts (\S\ref{sec:results:cascade}), \name takes the outages of step \stepA from three \emph{threat models}, one for the attack of \S\ref{sec:results:bipartite} and one for each island rule of \S\ref{sec:results:cascade}.
\begin{enumerate}[label=\textbf{T\arabic*}, leftmargin=*, nosep]
\item \emph{Attacker.} An attacker disables the $B_a$ substations whose loss takes the most IXP presence offline, as in the physical attacks of Table~\ref{tab:bg:taxonomy}.
Following the defender-attacker formulation~\cite{brown2006defending}, the attacker chooses these substations after the hardening is in place.
\item \emph{Largest-island cascade.} The 1, 2, 5, or 10\% of backbone substations with the highest grid betweenness fail (Figure~\ref{fig:cascade_band_main}).
Only the largest island keeps power, as in the collapses of Italy in 2003 and Iberia in 2025~\cite{ucte2004italy,entsoe2026iberian}.
\item \emph{Generation-aware cascade.} After the same seeds fail, every island whose dispatchable generation covers its load keeps power ($\beta = 1$ in \S\ref{sec:results:cascade}), as in the 2006 split~\cite{ucte2007split}.
\end{enumerate}

At step \stepB, an operator can \emph{back up a substation}, for example with a pre-positioned mobile generator~\cite{lei2018mobile}.
A backed-up substation stays energized for 48~hours, the longest backup of the failure analysis (\S\ref{sec:results:duration}), even when the rest of its island loses power.
Every facility that has this substation among its candidates keeps power for these 48~hours.
We assume that a backed-up substation also survives the attacker of T1, which is optimistic for an attack that destroys equipment.
At step \stepC, the operator can instead \emph{extend the backup of a facility} from the 12-hour fuel minimum to 48~hours, which protects only that facility.
When every facility holds the 12-hour minimum, both actions matter only for events of 12--48~hours (Proposition~\ref{prop:window}).
The minimum already bridges a shorter event, and a longer one also exhausts the 48~hours of the hardened units.
We therefore evaluate both actions for a 16-hour event, the duration of the Iberian blackout in Spain (\S\ref{sec:results:duration}).

At step \stepD, the loss is the IXP presence at the offline facilities, averaged over the four seeds under T2 and T3.
\name can add to it the \emph{stranded} IXP presence, whose traffic fails over to a nearby site that the same event takes down (\S\ref{section:routing_results}).
At step \stepE, \name selects the $B_d$ units that lower the loss the most.
Each unit is an existing substation or facility, and $B_d$ is the \emph{hardening budget}.
Against a cascade, whose outages are fixed, selecting substations is a maximum coverage problem, because a backed-up substation covers every facility that has it as a candidate.
Adding at each step the substation that protects the most additional IXP presence reaches at least $1-1/e$ (63\%) of the optimal protection~\cite{nemhauser1978analysis}.
For facility backup, where each facility protects only itself, this greedy selection is optimal (Proposition~\ref{prop:coverage}).
Against the attacker, \name must also find the worst attack for each hardening.
This problem is NP-hard, because it contains the densest-subgraph problem of fixed size (Proposition~\ref{prop:attacker}).
Because a greedy attacker underestimates the loss by up to 94.8\% (Appendix~\ref{sec:shield:solver}), \name computes the attack exactly with a mixed-integer program.
At each greedy step, \name considers only the candidates of the facilities that this attack takes offline, because only these substations can lower its loss.

\label{sec:model:properties}

\label{sec:shield:setup}

\subsection{Evaluation}
\label{sec:shield:results}
We evaluate \name on the 1{,}392 IXP presences at the 550 facilities of the map (\S\ref{sec:method:join}), with budgets $B_d$ from 0 to 30 units.
Unless stated otherwise, the attacker of T1 disables ten substations ($B_a = 10$).
An exact solver finishes for budgets of up to five units under T1 and up to 30 under T2 and T3 (Table~\ref{tab:shield_gap}).
At these budgets, the greedy selection is optimal under T3, within 2.3\% of the optimum under T2, and within 11 IXP presences under T1.
As baselines, we take five rankings that an operator can compute without \name and report the best of them at each budget.
Four of them rank substations by how many facilities they may feed (most-shared, \S\ref{sec:results:bipartite}), by grid betweenness (\S\ref{sec:results:cascade}), by betweenness in the dependency graph, and by the IXP presence they serve.
The fifth is a random order, averaged over 20 draws.

\emph{Above two substations, \name attains a lower loss than every ranking under all three threats} (Figure~\ref{fig:shield_curves}).
With ten backed-up substations, \name removes 72\% of the loss under T1, 35\% under T2, and 37\% under T3, against 52\%, 17\%, and 6\% for the best ranking (Table~\ref{tab:shield_headline_full}, Appendix~\ref{app:shield_more}).
The advantage of \name grows from a factor of 1.4 under T1 to 5.9 under T3.
Under T1, the best ranking places first the substations that many facilities share, which the attacker also targets.
Under T3, the loss depends on which islands lack generation, which the dependency and transmission graphs do not encode.

Under T1, the optimal attack on ten substations takes 167 IXP presences (12.0\%) offline without hardening, 121 of them in Amsterdam and 46 in Frankfurt.
The ten substations of \name include a candidate of every facility that this attack takes offline.
With them backed up, the optimal attack moves to Stockholm and Zurich, where it takes 47 IXP presences (3.4\%) offline and no IXP down.
In contrast, the most-shared order, the best ranking at this budget, has only its first two substations in common with \name.
Four of its other eight substations lie in Kyiv, Bucharest, and Sofia, far from the substations that the attack strikes.
The attack therefore still takes 80 IXP presences (5.7\%) offline, 30 of them in Frankfurt.
Ranking by grid betweenness leaves the loss at 12.0\% at every budget, because this ranking selects the substations at which cascades start.

Under T2, backing up the ten substations of \name lowers the mean loss from 39.8\% to 26.0\% of IXP presence, against 33.1\% for the best ranking.
At the 5\% seed (Figure~\ref{fig:islands_map}), 719 IXP presences go offline under the largest-island rule.
Of these, 259 lie in Amsterdam, Warsaw, Bucharest, and Madrid, where \name places its first four substations (Table~\ref{tab:shield_sets}).
With all ten backed up, the loss in these four metros falls to 46 IXP presences.
In contrast, four of the ten substations of the best ranking lie in Frankfurt, inside the Alpine island that keeps power under the largest-island rule.
Backing up the ten substations of this ranking keeps 143 IXP presences online, 120 of them in Amsterdam.
Under T3, \name reduces the mean loss from 12.4\% to 7.9\%, whereas the best ranking reduces it only to 11.7\%.
At the 5\% seed, the offline facilities lie in the islands around Paris, Berlin, Kyiv, Tallinn, and D\"usseldorf, which lack the generation to cover their load.
\name reduces the loss of these facilities from 203 to 125 IXP presences.
Under the best ranking, all of them stay offline, because six of its ten substations lie in Amsterdam and Frankfurt, which keep power under the generation-aware rule.

\label{sec:shield:where}
\label{sec:shield:robustness}

\begin{figure}[t]
\begin{minipage}[t]{0.47\linewidth}
\centering
\captionof{table}{Substations that \name selects in the first five greedy steps under each threat model, by the metro of the facilities they feed, with the IXP presence they serve in parentheses.}
\label{tab:shield_sets}
\footnotesize
\setlength{\tabcolsep}{2pt}
\begin{tabular}{@{}llll@{}}
\toprule
& \makecell[l]{T1\\attacker} & \makecell[l]{T2\\largest island} & \makecell[l]{T3\\generation-aware} \\
\midrule
1 & Amsterdam (104) & Amsterdam (104) & Frankfurt (43) \\
2 & Frankfurt (77) & Warsaw (30) & Paris (10) \\
3 & Stockholm (42) & Bucharest (27) & Tallinn (10) \\
4 & Frankfurt (35) & Madrid (19) & Kyiv (14) \\
5 & Kyiv (26) & Prague (14) & Paris (9) \\
\bottomrule
\end{tabular}
\end{minipage}\hfill
\begin{minipage}[t]{0.51\linewidth}
\centering
\vspace{0pt}
\includegraphics[width=\linewidth]{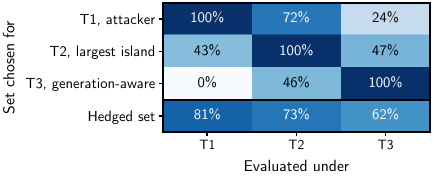}
\caption{Protection ratio at $B_d = 10$, the IXP presence that the set chosen for the threat of a row protects under the threat of a column, relative to the set chosen for the column. Off-diagonal ratios of 0--72\% arise because the threats take different facilities offline.}
\label{fig:shield_robustness}
\end{minipage}
\vspace{-8pt}
\end{figure}

Because the threats take different facilities offline, the first substations that \name selects lie in different metros under each threat (Table~\ref{tab:shield_sets}).
The \emph{protection ratio} of a set under a threat is the IXP presence it protects there, relative to the set chosen for that threat (Figure~\ref{fig:shield_robustness}).
Off the diagonal, the ratio is 0--72\%, and it is zero for the set chosen for T3 under T1.
We also select a \emph{hedged set} that minimizes its mean loss over the three threats, each relative to the set chosen for that threat.
This set keeps 62--81\% of the protection (Figure~\ref{fig:shield_robustness}, bottom row).
\emph{An operator therefore has to choose the threat before the substations, or accept the 19--38\% of protection that the hedged set gives up.}

\label{sec:shield:duration}

\begin{figure}[t]
\centering
\includegraphics[width=\linewidth]{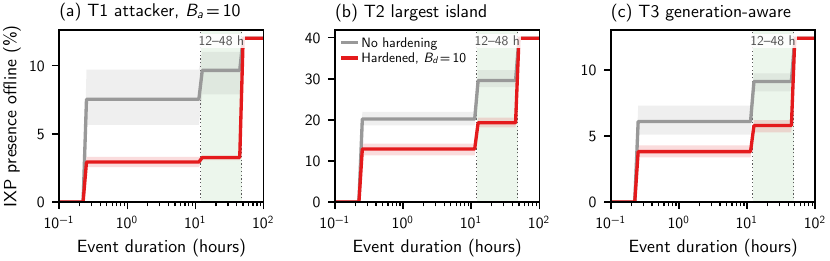}
\caption{IXP presence offline against the event duration $T$, without hardening and with the ten substations that \name selects for each threat, when facilities hold UPS only, 12~hours, or 48~hours of backup with probability 1/2, 1/4, and 1/4 (bands, 5th--95th percentile over 20 draws). The backed-up substations lower the loss up to 48~hours, when their own backup runs out.}
\label{fig:shield_duration}
\vspace{-8pt}
\end{figure}

Because facility backup is private, we assign each facility UPS only, 12~hours, or 48~hours with probability 1/2, 1/4, and 1/4 (Figure~\ref{fig:shield_duration}).
Under this assignment, a 2-hour event such as the 2006 split~\cite{ucte2007split} already takes IXP presence offline, because it outlasts the 15 minutes of UPS.
Backing up the ten substations of \name reduces the loss of this event by 36--61\%, because a backed-up substation keeps its facilities energized whatever their own backup.
A set chosen with the backup of each facility known protects only 4--10 more IXP presences than the set chosen for 12~hours everywhere.
\name can thus select its units without knowing the backup of each facility.

Per hardened unit, substation backup protects 1.3--2.0$\times$ as much IXP presence as facility backup, because one substation covers every facility that has it as a candidate (Figure~\ref{fig:shield_levers}).
If the attack destroys the backed-up substations it strikes, extending the backup of ten facilities still removes 54.5\% of the loss under T1.
Once a substation costs twice as much to harden as a facility under T1, or three times as much under T2 and T3, \name spends most of its budget on facilities (Table~\ref{tab:shield_cost}).
Because a selection over both actions protects at least as much as either alone, an operator who knows its cost ratio can leave this choice to \name.

To count stranded IXP presence, we group the facilities within 10~km of each PEERING site into a \emph{metro}.
A metro is down when at least half of its IXP presence is offline.
The stranded IXP presence of a down metro is the share of its IXP presence whose monitored addresses fail over to another down metro in the PEERING experiments.
This count is a lower bound, because the experiments locate the alternate site of only 23--92\% of the monitored addresses of a site.
At the 5\% seed of T2, Amsterdam, Paris, Warsaw, and Madrid are down at once.
Between 7\% and 40\% of the monitored addresses of each fail over to another of the four (Figure~\ref{fig:shield_failover}).
These shares give 45.6 stranded IXP presences, most of them from Warsaw and Paris, whose traffic fails over to Amsterdam.

An attacker who disables twenty substations to maximize offline and stranded IXP presence together takes 262 IXP presences offline.
It strands another 121.8, which raises its loss by 46\%.
Although its objective ignores stranded traffic, backing up the first five substations of \name removes all stranding under both cascades and against this attacker.
These five substations keep Amsterdam, Madrid, and Warsaw up at the 5\% seed of T2 and Amsterdam and Frankfurt up against the attacker.
Adding stranded traffic to the objective changes two of them when a metro counts as down at a quarter of its IXP presence, and at most one at a higher share (Table~\ref{tab:shield_threshold}).

\label{sec:shield:levers}

\label{sec:shield:routing}

\begin{tcolorbox}[title={Takeaway \rqbox{4}}]
An operator prevents the most IXP outages by backing up the substations that \name selects for the expected threat, unless a substation costs a few times as much to back up as a facility.
This selection depends on the threat, because attacks and grid splits take different facilities offline, and a set chosen for one threat loses at least a quarter of its protection under another threat, and in one case all of it.
\end{tcolorbox}

\section{Conclusion}
\label{section:conclusion}

\tool is, to our knowledge, the first study that maps the interdependence of the power grid and Internet infrastructure from the substation to the IXP and to its traffic.
We map 20{,}085 high-voltage substations from OpenStreetMap, link 552 PeeringDB facilities hosting 275 IXPs to their candidate feeders, and validate both against operator records in the Netherlands and France.
The dependency is concentrated, as half of European IXP presence loses power once the most-shared 16.2\% of these feeders fail, less than half of what the shuffled and rewired maps need.
Our PEERING measurements show that most monitored addresses of a failed metropolitan area move to nearby sites, which in continental Europe often share its synchronous grid.
Propagation causes most of the outage of a grid event, as losing the 5\% most central backbone substations cuts power to 0.1\% of IXP presence directly but to 8.8--51.7\% once the grid splits.
This loss of power becomes an outage only when the event outlasts the 12-hour fuel minimum of facility backup, as the Iberian blackout did in Spain.
Building on these findings, \name selects the substations or facilities to harden under a budget.
Against an adversary who disables ten substations, hardening the ten that \name selects removes 72\% of the IXP presence taken offline, whereas the best standard ranking removes 52\%.
Because a set chosen for one threat retains only 0--72\% of its protection under another, the hardening decision depends on the threat model that the operator adopts.

\label{section:limitations}
\label{sec:shield:limits}

\smallskip\noindent\textbf{Limitations and future work.}
First, because operators do not publish which substation feeds a facility, we take its three nearest substations as candidate feeders and validate them only against the primary feeders of 61 Dutch facilities.
A facility may, however, draw its power through several medium-voltage connections, and the accuracy of the candidate feeders outside the Netherlands remains to be tested.
Second, OpenStreetMap omits the inner-city 110~kV cable network of Frankfurt, and our cascades may therefore understate the outage there.
Third, our island rules decide whether an island keeps power from its topology and installed capacity, whereas a real island survives only if its inertia, reserves, and under-frequency load shedding keep its frequency within operating limits~\cite{ucte2007split}.
Fourth, IXP presence measures the loss of the sites where networks exchange traffic, but not the loss of demand when IXP members and their customers lose power.
This loss of demand may explain why traffic fell at the Iberian IXPs whose facilities remained on backup power.
Finally, our PEERING experiments withdraw a prefix from an entire metropolitan area at once, whereas its facilities fail at different times as their backup runs out.
The catchment of our anycast prefix only approximates how the members of a failed IXP reroute the traffic they exchange with one another.

We will validate the candidate feeders beyond the Netherlands with records from grid and facility operators, and add inner-city cable networks where operators publish them.
To predict which islands survive a split, we will couple the map with a dynamic model of the European grid run on the dispatch of documented events.
On PEERING, we will fail sites as a grid split exhausts their backup, and combine the results with IXP traffic data to quantify the traffic that a grid event displaces.
Because \tool requires only open data, its map can be rebuilt wherever OpenStreetMap maps the transmission grid.
We hope that the release of our dataset and code starts a community effort to measure what this study had to infer.

\newpage
\bibliographystyle{unsrt}
\bibliography{references}

\newpage

\appendix

\section{Use of Generative AI}\label{app:genai}

We used generative AI tools to improve the grammar and readability of this paper.
All technical content, analysis, and interpretation are entirely our own.

\section{Background Material}
\label{app:part_background}

\subsection{Backup Power at Facilities}
\label{app:backup}

Table~\ref{tab:bg:backup} lists the two stages of facility backup and the autonomy that standards and operators state for each.
The Uptime Institute Tier Standard requires 12~hours of fuel at design load at every Tier~\cite{uptime_tier_topology,heslin2014fuel}, because a higher Tier adds redundant components and paths rather than a longer reserve.
The standard leaves the UPS runtime to the operator~\cite{heslin2017myths}.
The European standard EN~50600-2-2 also defines its availability classes by redundancy, and its certification criteria require fuel storage and a delivery contract without stating hours~\cite{en50600_2_2,cencenelec2021gdc,tuvnord2024en50600}.
Because 12~hours is the only fuel reserve that these standards state, the failure analysis takes it as the minimum backup (\S\ref{sec:results:duration}).
As the longest backup, it takes 48~hours, the most that one hyperscale operator states~\cite{microsoft2022fuel}.
\begin{table}[!htbp]
\centering
\small
\caption{Backup power stages at data center facilities and the autonomy stated for each. Autonomy beyond the fuel stored on site depends on deliveries, which a wide-area grid outage may disrupt.}
\label{tab:bg:backup}
\begin{tabularx}{\linewidth}{@{}l >{\raggedright\arraybackslash}X l l@{}}
\toprule
Stage & Function & Autonomy & Source \\
\midrule
UPS, battery & Carries the load until generators start & 5--10~min, worst case 15~min & \cite{avelar2016battery,brown2019liion} \\
UPS, diesel rotary & Same, without batteries & 10--30~s & \cite{heslin2017myths} \\
On-site fuel, Tier I--IV minimum & Runs the generators at design load & 12~h & \cite{heslin2014fuel} \\
On-site fuel, EN~50600-2-2 & Same & Not stated in hours & \cite{tuvnord2024en50600} \\
On-site fuel, operator practice & Same & Up to 48~h at one hyperscale operator & \cite{microsoft2022fuel} \\
Refueling contract & Refills the tanks during the outage & Beyond on-site fuel, if deliveries arrive & \cite{davis2023grid} \\
\bottomrule
\end{tabularx}
\end{table}

\subsection{Failure Orderings and Documented Grid Events}
\label{app:taxonomy}

Table~\ref{tab:bg:taxonomy} pairs each rule by which the failure analysis selects the first substations to fail (\S\ref{section:grid_results}) with grid events that official investigations document.
Because we model which substations fail first and not why they fail, the pairing follows these first failures rather than the physical cause of an event.
Random failures stand for events that do not depend on the role of a substation in the grid, such as lightning and cold weather.
Most-shared removal (\S\ref{sec:results:bipartite}) stands for the physical attacks that disable a few substations directly, as at Metcalf and in Ukraine.
Seeds ranked by grid betweenness, and corridor seeds, stand for the loss of the corridors that hold a synchronous area together, as in the splits of 2003 and 2006 (\S\ref{sec:results:cascade}).
The loss of every substation of a country, as we seed the Iberian blackout (\S\ref{sec:results:duration}), stands for an event that de-energizes a whole national grid.
\begin{table}[!t]
\centering
\footnotesize
\caption{Rules by which the failure analysis selects the first substations to fail, each with documented grid events that match it. Scale and duration are those of the cited investigation report.}
\label{tab:bg:taxonomy}
\setlength{\tabcolsep}{3pt}
\begin{tabularx}{\linewidth}{@{}>{\raggedright\arraybackslash}p{1.55cm} >{\raggedright\arraybackslash}p{1.9cm} >{\raggedright\arraybackslash}X >{\raggedright\arraybackslash}p{3.1cm} >{\raggedright\arraybackslash}p{2.6cm} l@{}}
\toprule
Selection & Event & Initiating cause & Scale & Duration & Source \\
\midrule
Random & Great Britain, 9 Aug 2019 & Lightning strike on a line, then loss of two generators & 1.15~M customers, 892~MW & 15--45~min & \cite{ofgem2020august,eso2019technical} \\
Random & Texas, Feb 2021 & Generators fail from freezing and loss of fuel supply & 20{,}000~MW peak load shed, $>$4.5~M people & Load shed for nearly 3~days, some customers 4~days & \cite{ferc2021uri} \\
\addlinespace
Most-shared & Metcalf, 16 Apr 2013 & Rifle attack on one transmission substation & 17 transformers damaged, no customer lost power & Substation out for about one month & \cite{brinkman2015physical} \\
Most-shared & Ukraine, 23 Dec 2015 & Remote opening of breakers at three distribution companies & $\geq$27 substations, 225{,}000 customers & Several hours & \cite{lee2016ukraine} \\
Most-shared & Ukraine, from Oct 2022 & Missile and drone attacks on substations and generation & About half of the large network substations damaged in 2022--23 & Supply of a few hours per day in the worst-affected regions & \cite{iea2024ukraine} \\
\addlinespace
Grid betweenness, corridors & UCTE, 4 Nov 2006 & Switch-off of a 380~kV double circuit, then cascading line trips & CESA split in three, 15~M households, 17{,}000~MW shed & Resynchronized in 38~min, normal in $<$2~h & \cite{ucte2007split} \\
Grid betweenness, corridors & Italy, 28 Sep 2003 & Trips of the Lukmanier and San Bernardino 380~kV lines & Italy except Sardinia, 27{,}444~MW load, 177~GWh not supplied & Last customer after 18~h 12~min & \cite{ucte2004italy} \\
\addlinespace
Whole country & Iberia, 28 Apr 2025 & Overvoltage, then cascading loss of generation & Spain and Portugal, 31{,}538~MW, 100\% of demand & Grid restored after 12~h (PT) and 16~h (ES) & \cite{entsoe2026iberian} \\
\bottomrule
\end{tabularx}
\end{table}

\section{The Map in Detail}
\label{app:part_map}

\subsection{Extraction Steps in Detail}
\label{app:extraction_steps}
\label{app:extraction_funnel}

Five steps reduce the OSM objects of Europe to 24{,}900 substations, 20{,}085 of them in the scope.

\smallskip\noindent\textbf{Step 1. Data acquisition.} We ingest every \code{power=substation} object from OSM across Europe and every \code{power=line} and \code{power=cable} polyline tagged at $\geq$110~kV.
Because substations appear as points, polygons, or multi-polygon relations, we replace each polygon with a representative interior point.
We convert the voltage tags to kilovolts, and a multi-value tag such as \code{380000;150000} takes its maximum, here 380~kV.

\smallskip\noindent\textbf{Step 2. Subtype exclusion.} We exclude the eight subtypes that the OSM tagging convention defines as non-transmission (\code{minor\_distribution}, \code{kiosk}, \code{traction}, \code{production}, \code{measurement}, \code{compensation}, \code{valve}, \code{historic})~\cite{osmwiki_substation}.
We keep three further subtypes (\code{industrial}, \code{generation}, \code{converter}) only at $\geq$220~kV, where they are transmission assets such as HVDC converter stations.

\smallskip\noindent\textbf{Step 3. Quality gate.} A candidate tagged at $\geq$110~kV skips to Step~5.
A named candidate without a voltage tag goes to Step~4 unless its subtype is \code{distribution}, and the gate discards every other candidate.
Of the 692{,}387 candidates that remain after Step~2, the gate keeps 23{,}736 tagged and 15{,}748 named candidates.
The gate discards 636{,}053 unnamed objects, 15{,}932 named objects tagged below 110~kV, and 918 named \code{distribution} objects without a voltage tag.

\smallskip\noindent\textbf{Step 4. Proximity validation.} Because a name alone does not show that a station belongs to the transmission grid, we keep a named candidate only if it lies within 1{,}000~m of a line at $\geq$220~kV or within 500~m of a line at 110--220~kV.
This test keeps 3{,}928 of the 15{,}748 named candidates.

\smallskip\noindent\textbf{Step 5. Deduplication and voltage assignment.} Because OSM often maps one substation as several features, we merge features that lie close together with DBSCAN at $\epsilon$=300~m and \code{min\_samples}=1.
Each station lies at the centroid of its cluster.
A larger radius would absorb yards whose features lie farther apart, but it also merges distinct stations.
Against TenneT, recall at a 500~m match falls from 96.4\% at 300~m to 77.8\% at 2{,}000~m (Figure~\ref{fig:validation_combined}b), and we therefore keep 300~m, the smallest radius in the sweep.
Because we compute no power flows, each station takes the highest voltage among its features, the membership rule of GridKit at 220~kV~\cite{gridkit2017}.
Stations from Step~4 are capped at 110~kV, because proximity to a line does not establish their voltage.

\begin{figure}[!t]
\centering
\includegraphics[width=\linewidth]{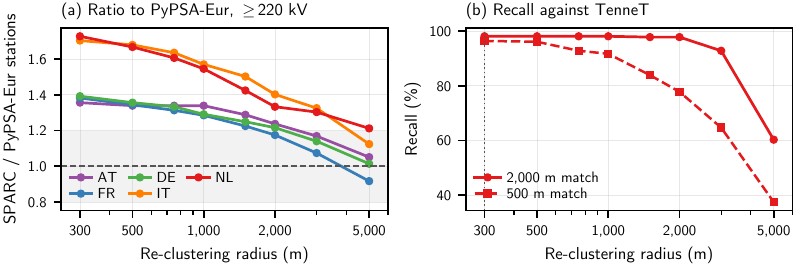}

\caption{Re-clustering the 300~m output of Step~5 at larger radii, with (a) the ratio of \tool to PyPSA-Eur stations at $\geq$220~kV in the five overcounting countries (grey band, 0.8--1.2) and (b) recall against TenneT at two match radii. The overcount shrinks smoothly, as bus-section duplication predicts, but recall falls once distinct stations merge, which is why Step~5 keeps 300~m (dotted).}
\label{fig:validation_combined}
\end{figure}

\subsection{Validation in Detail}
\label{app:validation_details}

\smallskip\noindent\textbf{Match radius and chance baselines.}
To test whether the TenneT match depends on the 2~km radius or on the density of the \tool stations, we repeat the match at four radii and against two chance baselines (Table~\ref{tab:validation_nulls}).
The first baseline shifts every \tool station by 5~km in one of eight directions and averages the result over the eight.
The second replaces the \tool stations with 421 points drawn uniformly at random inside their convex hull in the Netherlands.
Across the radii, recall varies by only 2.5 percentage points, whereas both baselines stay at or below 15.0\%, which shows that the match reflects where \tool places the stations.
Precision rises with the radius, from 72.9\% at 250~m to 86.2\% at 2~km.
At 2~km, recall is at least 99.5\% in each of the four voltage tiers from 110 to 380~kV.

\smallskip\noindent\textbf{Stations without a TenneT match.}
Because \tool resolves one station per yard and TenneT lists a site with several yards once, the match lets several \tool stations recover the same TenneT station.
Of the 58 \tool stations that match no TenneT station, 40 are points where a TenneT line passes from cable to overhead, which OSM tags as \code{substation=transition}, and 14 are named stations at $\geq$110~kV that the register omits, such as wind park connections.
Dropping the transition points that form their own cluster would raise precision to 94.0\% at unchanged recall.
We report this rule as a correction for future releases and do not apply it, because it would change the 24{,}900 substations and every number that depends on them.

\begin{table}[!htbp]
\centering
\small
\caption{TenneT station match by radius, for \tool and two chance baselines. Values are percentages of the 360 TenneT stations (recall) and of the matched point set (precision).}
\label{tab:validation_nulls}
\begin{tabular}{lrrrr}
\toprule
 & \multicolumn{2}{c}{\tool} & \multicolumn{2}{c}{Chance baseline, recall} \\
\cmidrule(lr){2-3}\cmidrule(lr){4-5}
Radius & Recall & Precision & Shifted 5~km & Random points \\
\midrule
250~m & 95.6 & 72.9 & 0.2 & 0.0 \\
500~m & 96.4 & 75.3 & 1.6 & 0.3 \\
1~km & 97.2 & 79.8 & 4.6 & 2.2 \\
2~km & 98.1 & 86.2 & 15.0 & 10.6 \\
\bottomrule
\end{tabular}
\end{table}

\smallskip\noindent\textbf{Dutch facilities in the feeder test.}
Three Dutch facilities share postcode 1098\,XH at Amsterdam Science Park, which the capacity map assigns to a station 6~km away.
The 150~kV Watergraafsmeer station, however, stands inside this postcode~\cite{arcam_sciencepark} and feeds the five other 1098\,X postcodes.
The feeder test of \S\ref{sec:method:validation} therefore excludes the three facilities.
If the test kept them with the station that the map names, their eleventh-nearest substation, the candidate set with $k=3$ would contain the feeder for 58 of 64 facilities (90.6\%).

\smallskip\noindent\textbf{PyPSA-Eur station counts.}
PyPSA-Eur covers every scope country except Liechtenstein, and over these 33 countries it lists 3{,}324 stations at $\geq$220~kV against 3{,}876 for \tool.
We measure agreement with Pearson $r$ on the counts and with Spearman $\rho$ on the country ranks, which the few large countries dominate less.
The counts agree with $r=0.983$ and $\rho=0.979$, and the median ratio of \tool to PyPSA-Eur is 0.97, with 20 of the 33 countries within 0.8--1.2.
The five countries with the most facilities, which hold 287 of the 552, all exceed the PyPSA-Eur count, from a ratio of 1.13 in Poland to 1.73 in the Netherlands, whereas Ukraine, where OSM coverage is thinnest, has the lowest ratio of 0.52.
The overcount is largest in Austria, France, Germany, Italy, and the Netherlands, where \tool extracts 1.36--1.73$\times$ the PyPSA-Eur count.
In bus-section duplication, OSM maps the bus sections of one substation as separate objects.
We attribute the excess to this duplication, because the overcount shrinks smoothly in each of these countries as the DBSCAN radius grows (Figure~\ref{fig:validation_combined}a).
Because a larger radius also merges distinct stations (Step~5), we keep 300~m and accept the overcount, which adds duplicate points at real yards rather than false stations.
Even at 5~km, where recall against TenneT has fallen to 60.3\%, the ratio stays above 1 in four of the five countries, which points to stations missing from PyPSA-Eur.
Conversely, 18 of the 33 Dutch PyPSA-Eur substations (54.5\%) lie within 2~km of a TenneT substation, at a median distance of 1.7~km.

\smallskip\noindent\textbf{France against RTE.}
Of the 1{,}066 RTE substations, 167 match no \tool station.
The extraction removes 98 of them by design, because they are traction (59), generation (21), distribution (11), or industrial (4) sites or lie below the voltage threshold (3).
Another 38 carry codes absent from OSM, although spot checks find some of them under variant codes, such as \code{COR.P} for \code{CORD5}.
The remaining 31 are lost when Step~5 merges them with a neighbouring station or when the clipping to the scope removes them.

\label{app:validation_breadth}

\subsection{Substations Nearest to the Most IXP Presence}
\label{app:hotspots}

To find where IXP presence concentrates, we rank substations by the IXP presence for which they are the nearest high-voltage source.
The top 44 of them are the nearest source for half of all 1{,}394 IXP presences.
The loss of this nearest source alone does not darken a facility, because a facility goes dark only when all three of its candidates fail (\S\ref{section:grid_results}).

\section{PEERING Experiments in Detail}
\label{app:part_peering}

\subsection{PEERING Deployment Details}\label{app:peering}

Each announcement configuration stays active for 90~minutes, which avoids route flap dampening and leaves enough time to measure catchments and routes.
To measure the latency from a site, we announce a prefix from that site alone and send three pings from our router to each of the 404{,}957 monitored addresses.
Continuous pings between our router and the site give the baseline latency, which we subtract from the end-to-end latency to each address.
From the RIPE Atlas traceroutes, we keep a route only when the last three traceroutes from a probe to a prefix during an announcement traverse the same sequence of ASes.

\subsection{Additional Discussion of Outage Experiments}\label{app:peering.results}

Grouping the monitored addresses by AS (Figure~\ref{fig:app.peering.weighting}a) or weighting each AS by its estimated traffic volume (Figure~\ref{fig:app.peering.weighting}b) gives reroute distances similar to those of individual addresses (Figure~\ref{fig:peering.isolated}a).
We estimate the traffic volume of an AS as proportional to its population in the APNIC AS-population dataset~\cite{apnicaspop}.
Over all sites, 44\% of the rerouted addresses, 44\% of the ASes, and 41\% of the estimated volume move less than 1{,}000~km from the failed site.

\begin{figure}[!t]
\centering
\includegraphics[width=\linewidth]{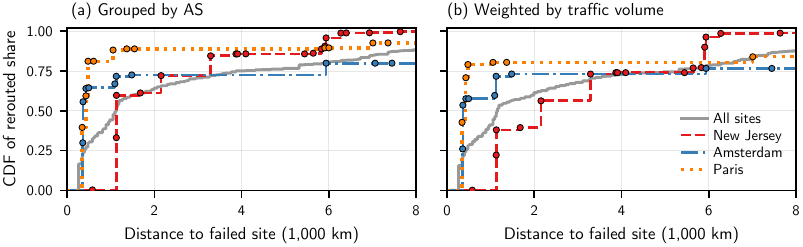}
\caption{Reroute distance after isolated outages, with the monitored addresses (a) grouped by AS and (b) weighted by the traffic volume of their AS, estimated from the APNIC AS-population dataset~\cite{apnicaspop}. Both distributions resemble the distribution over individual addresses in Figure~\ref{fig:peering.isolated}a.}
\label{fig:app.peering.weighting}
\end{figure}

For the isolated outage of each Vultr site, we report the share of the monitored addresses in its catchment that move to each of its four nearest alternate sites (Table~\ref{tab:app.peering.results}).
Over the 29 outages, these four sites receive a median of 67.2\% of the addresses, and at least half of them in 23 outages.

In 12 of the 29 outages, one of the three farther sites receives more of these addresses than the nearest site, as in the outage of New Jersey.
Toronto, the nearest alternate of New Jersey at ${\sim}$550~km, receives 0.5\% of its addresses, whereas Atlanta and Chicago, ${\sim}$1{,}150~km away, receive ${\sim}$27\% each.
Some addresses also cross an ocean, as 12.8\% of those that New Jersey serves move to Amsterdam, Frankfurt, and Madrid, ${\sim}$5{,}900~km away.
These addresses often belong to well-connected service and hosting providers, such as GoDaddy and Fly.io, and to ISPs that are customers of large Tier-1 networks, such as Neterra (AS34224) and Skyway West (AS14007).
Fiberway (AS49434) takes a regional-transit detour like those through EdgeUno in \S\ref{section:routing_results}, because it exchanges traffic with PEERING directly in Paris but reaches the Mumbai site through Airtel when Paris fails.
In the cascades from Amsterdam, more than a quarter of the catchment sees lower latency than at Amsterdam for the first three failures (Figure~\ref{fig:peering.cascading}b).
Of the catchment of Amsterdam, 54\% lies more than 2{,}000~km away, in regions such as the Middle East, India, Russia, and Turkey, and these addresses may reach Amsterdam through remote peering at AMS-IX and move to a closer site once it fails.
European sites also serve distant addresses, as 32\% of the addresses that prefer Frankfurt (DE-CIX) and 27\% of those that prefer London (LINX) lie more than 2{,}000~km away.

\begin{table}[!t]
    \centering
\resizebox{\textwidth}{!}{
\begin{tabular}{l|llllllll}
\toprule
\multirow{1}{*}{} & \multicolumn{2}{c}{Nearest \#1} & \multicolumn{2}{c}{Nearest \#2} & \multicolumn{2}{c}{Nearest \#3} & \multicolumn{2}{c}{Nearest \#4} \\
Outage Site & Site & \% IPs & Site & \% IPs & Site & \% IPs & Site & \% IPs \\
\midrule
Amsterdam & London & 27.0\% & Frankfurt & 21.8\% & Paris & 6.9\% & Manchester & 0.4\% \\
Atlanta & Chicago & 4.3\% & Miami & 10.2\% & New Jersey & 41.7\% & Dallas & 13.3\% \\
Chicago & Toronto & 13.7\% & Atlanta & 18.2\% & New Jersey & 37.9\% & Dallas & 8.5\% \\
Dallas & Atlanta & 60.6\% & Chicago & 12.4\% & Mexico & 5.5\% & Miami & 3.3\% \\
Frankfurt & Amsterdam & 30.8\% & Paris & 16.3\% & London & 2.2\% & Manchester & 0.0\% \\
Honolulu & Silicon Valley & 13.5\% & Los Angeles & 0.7\% & Seattle & 32.6\% & Mexico & 15.6\% \\
Johannesburg & Tel Aviv & 0.0\% & São Paulo & 4.6\% & Madrid & 4.2\% & Singapore & 0.3\% \\
London & Manchester & 15.7\% & Paris & 11.1\% & Amsterdam & 17.5\% & Frankfurt & 9.0\% \\
Los Angeles & Silicon Valley & 12.5\% & Seattle & 32.0\% & Dallas & 18.1\% & Mexico & 13.7\% \\
Madrid & Paris & 17.2\% & London & 3.2\% & Frankfurt & 2.7\% & Manchester & 0.3\% \\
Manchester & London & 92.2\% & Amsterdam & 0.9\% & Paris & 0.0\% & Frankfurt & 0.0\% \\
Melbourne & Sydney & 94.9\% & Singapore & 2.8\% & Osaka & 0.0\% & Tokyo & 0.0\% \\
Mexico & Dallas & 2.3\% & Miami & 0.3\% & Atlanta & 1.5\% & Los Angeles & 1.2\% \\
Miami & Atlanta & 56.9\% & New Jersey & 6.5\% & Dallas & 4.1\% & Chicago & 0.5\% \\
New Jersey & Toronto & 0.5\% & Chicago & 27.6\% & Atlanta & 27.3\% & Miami & 1.2\% \\
Osaka & Tokyo & 96.9\% & Seoul & 0.6\% & Singapore & 0.0\% & Honolulu & 0.0\% \\
Paris & London & 31.4\% & Amsterdam & 24.0\% & Frankfurt & 24.4\% & Manchester & 0.0\% \\
Santiago & São Paulo & 98.1\% & Mexico & 0.4\% & Miami & 0.7\% & Atlanta & 0.1\% \\
São Paulo & Santiago & 54.0\% & Miami & 6.1\% & Johannesburg & 0.0\% & Mexico & 2.0\% \\
Seattle & Silicon Valley & 28.9\% & Los Angeles & 23.2\% & Dallas & 0.3\% & Chicago & 17.6\% \\
Seoul & Osaka & 22.7\% & Tokyo & 26.0\% & Singapore & 1.5\% & Honolulu & 0.0\% \\
Silicon Valley & Los Angeles & 17.5\% & Seattle & 1.4\% & Dallas & 2.6\% & Chicago & 0.1\% \\
Singapore & Seoul & 1.3\% & Osaka & 0.9\% & Tokyo & 29.9\% & Melbourne & 2.4\% \\
Stockholm & Warsaw & 1.1\% & Amsterdam & 22.5\% & Frankfurt & 68.4\% & Manchester & 0.0\% \\
Sydney & Melbourne & 55.4\% & Singapore & 9.9\% & Osaka & 0.0\% & Tokyo & 1.9\% \\
Tel Aviv & Warsaw & 0.0\% & Frankfurt & 57.7\% & Stockholm & 0.0\% & Paris & 5.5\% \\
Tokyo & Osaka & 44.2\% & Seoul & 11.1\% & Singapore & 6.4\% & Honolulu & 0.0\% \\
Toronto & New Jersey & 78.7\% & Chicago & 13.1\% & Atlanta & 0.4\% & Dallas & 0.1\% \\
Warsaw & Stockholm & 2.0\% & Frankfurt & 34.3\% & Amsterdam & 40.3\% & Paris & 0.0\% \\
\bottomrule
\end{tabular}
}
    \caption{Share of the monitored addresses of each failed site that move to each of its four nearest alternate sites in isolated outages, ordered by distance.}
    \label{tab:app.peering.results}
\end{table}

\section{Failure Analysis in Detail}
\label{app:part_failure}

\subsection{Outage Rule and Direct Loss of Shared Substations}
\label{app:calibrated}

\smallskip\noindent\textbf{Outage rule.}
Because the substation that supplies a facility is private to the facility and its grid operator, the outage rule of \S\ref{section:grid_results} calls a facility dark only when all its candidates are de-energized.
When the candidate set contains the substation that actually supplies the facility, every facility that this rule calls dark has lost power.
Adding candidates can only lower the IXP presence at dark facilities, $\Phi$ in Eq.~\eqref{eq:ixp_offline}.

\begin{proposition}[Conservative outage rule]
\label{prop:conservative}
Let $\phi(f)$ be the true feeder of facility $f$.
If $\phi(f) \in N(f)$, then $f$ is dark under the outage rule only for sets $X$ that contain $\phi(f)$.
If two candidate sets satisfy $N(f) \subseteq N'(f)$ for every $f$, then $\Phi'(X) \le \Phi(X)$ for every $X$.
\end{proposition}

\begin{proof}
If $f$ is dark, then $N(f) \subseteq X$, and $\phi(f) \in N(f)$ gives $\phi(f) \in X$.
For the second statement, $N'(f) \subseteq X$ implies $N(f) \subseteq X$ when $N(f) \subseteq N'(f)$.
Every facility that is dark under $N'$ is therefore dark under $N$, and because the weights are the same, $\Phi'(X) \le \Phi(X)$.
\end{proof}

\smallskip\noindent\textbf{Ties in the most-shared order.}
The most-shared order breaks ties by the IXP presence served and then by OSM identifier.
The tie rule sets much of the order, because at $k=3$, 661 of the 984 facility-connected substations are candidates of one facility and 193 of two.
The threshold of 159 substations falls among those with two facilities, whose order the tie rule therefore decides.
Breaking ties at random instead moves the threshold to 21.3--27.9\% in Europe and 11.2--17.8\% in the Netherlands over 200 random orders.
The order with ties by IXP presence, at 16.2\% and 11.2\%, lies at or below the low end of both ranges.
To cut power to a quarter and three quarters of European IXP presence, this order removes 58 and 402 substations, against 618 and 892 for random removal.

\smallskip\noindent\textbf{Shuffled and rewired maps.}
To test whether the concentration comes from where facilities and their IXPs sit, we repeat most-shared removal on two altered maps of 1{,}000 draws each (Table~\ref{tab:nulls_concentration}).
The \emph{shuffled map} keeps the bipartite dependency graph and shuffles the IXP presences over the 550 facilities.
The \emph{rewired map} keeps the weights and applies 10 double-edge swaps per edge, which preserve the degree of every facility and every substation but not their geography.
This map therefore tests whether the degree sequence alone produces the concentration.
The random threshold is the same 79.4\% on the measured map and on both altered maps, because it depends only on the facility degrees and weights.
The 79.3\% of \S\ref{sec:results:bipartite} estimates this threshold from 1{,}000 random orders.
Each of the 2{,}000 draws needs more substations than the 159 of the measured map, because on the measured map the facilities that host the most IXPs share their candidate substations.

\begin{table}[!t]
\centering
\small
\caption{Concentration of the measured $k=3$ map against the shuffled and rewired maps of 1{,}000 draws each. Thresholds are the number of the 984 facility-connected substations removed to take 50\% of IXP presence offline. The columns of the altered maps give the mean and the 5th--95th percentile. Every draw is less concentrated than the measured map in every row ($p<0.001$).}
\label{tab:nulls_concentration}
\begin{tabular}{@{}lrrr@{}}
\toprule
Statistic & Measured & Shuffled & Rewired \\
\midrule
Most-shared threshold, substations & 159 & 331 (324--340) & 387 (377--397) \\
Most-shared threshold, \% of 984 & 16.2 & 33.6 (32.9--34.6) & 39.3 (38.3--40.3) \\
Ratio of random to most-shared threshold & 4.9 & 2.4 (2.3--2.4) & 2.0 (2.0--2.1) \\
Nearest source for half of IXP presence, substations & 44 & 60 (56--64) & -- \\
\bottomrule
\end{tabular}
\end{table}

\smallskip\noindent\textbf{Calibrated single-feeder rule.}
Because a facility goes dark under the outage rule only when all three of its candidates fail, the rule understates the loss whenever the actual feeder fails alone (Proposition~\ref{prop:conservative}).
To bound this effect, we give each facility one feeder, drawn among its five nearest substations within 15~km with the rank probabilities of the Dutch feeder test (39, 11, 8, 1, and 2 of 61 facilities for ranks 1 to 5).
For the 27 facilities with fewer than five candidates, we renormalize the probabilities over the available ranks.
We compare the outage rule with three candidates against the outage rule with one candidate and against 200 draws of the calibrated rule (Figure~\ref{fig:calibrated_feeder}).
In every setting, the least IXP presence goes offline with three candidates.
When the 984 substations are removed in the most-shared order of the three-candidate graph, half of IXP presence loses power after 12.5\% of them under the calibrated rule.
Over the draws, this share ranges from 10.8\% to 14.0\%, against 16.2\% under the outage rule.
Random removal needs 51.3\%, 4.1 times as many substations as most-shared removal, against 4.9 times under the outage rule.
The draws are independent across facilities, although co-located facilities probably share a feeder.
The percentile ranges therefore show how the draws vary, not the uncertainty of rank probabilities estimated from 61 facilities.

\label{app:national}

\begin{figure}[!t]
\centering
\begin{minipage}[c]{0.5\linewidth}
\centering
\includegraphics[width=\linewidth]{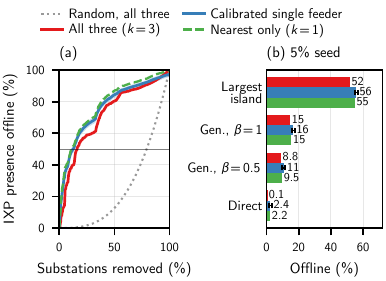}
\end{minipage}\hfill
\begin{minipage}[c]{0.47\linewidth}
\caption{IXP presence offline with all three candidates ($k=3$), with the nearest only ($k=1$), or with one calibrated feeder per facility, (a) under most-shared removal of the 984 facility-connected substations and (b) at a 5\% betweenness seed under three island rules and the direct model, with 5th--95th percentiles over 200 draws. In every setting, the least IXP presence goes offline with all three candidates.}
\label{fig:calibrated_feeder}
\end{minipage}
\vspace{-8pt}
\end{figure}

\begin{figure}[!t]
\centering
\includegraphics[width=0.9\linewidth]{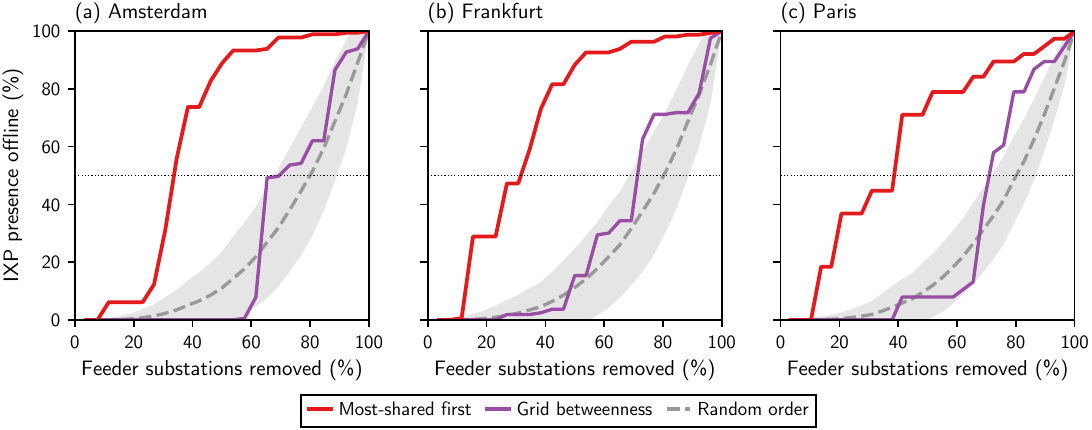}
\caption{IXP presence offline in Amsterdam, Frankfurt, and Paris as each metro's feeder substations are removed. Removing the most-shared substations first takes half of the metro's IXP presence offline after 35--41\% of them, whereas removal by grid betweenness tracks random order (band, $\pm 1$ standard deviation).}
\label{fig:app_robustness_cascade_metros}
\end{figure}

\smallskip\noindent\textbf{Per-metro view.}
To test whether the concentration holds within a metro, we repeat most-shared removal in Amsterdam, Frankfurt, and Paris, the three largest peering metros by facility count (Figure~\ref{fig:app_robustness_cascade_metros}).
In each metro, we remove only the candidate substations of its facilities.
We compare the most-shared order with removal by grid betweenness and with random removal.
In Frankfurt, removing 9 of its 26 candidate substations (35\%) in the most-shared order takes half of the metro's IXP presence offline, whereas random removal needs 81\%.
In Amsterdam, the two orders need 35\% and 81\% as well.
Removal by grid betweenness cuts power no faster than random removal in any of the three metros.
In Frankfurt, this removal has taken 1.8\% of the IXP presence offline when the most-shared order reaches half.
Within these metros, the candidates that many facilities share thus differ from those most central to the transmission grid.

\subsection{Spread Through Grid Islands}
\label{app:cascade_variants}

\smallskip\noindent\textbf{Composition of the backbone component.}
\label{app:backbone}
The largest connected component $C^\ast$ of the transmission graph, our proxy for the continental synchronous backbone, reaches beyond the scope.
It includes 1{,}528 substations in Great Britain, because OSM maps the HVDC interconnectors as power lines, and 1{,}434 in Belarus and Russia, through the lines to the Baltic states and Ukraine.
Today, a cascade stops at both links.
An HVDC link transfers power without passing on the frequency disturbances through which a cascade spreads.
The lines to Belarus and Russia have carried no power since Ukraine and the Baltic states disconnected from that grid in 2022 and 2025.
Another 33 substations of $C^\ast$ lie in Kosovo, Morocco, and Ireland.
Inside the scope, $C^\ast$ misses Sweden and Finland entirely and nearly all of Norway, Denmark, Latvia, and Bulgaria.
Sweden, Finland, Norway, and eastern Denmark form the Nordic synchronous area, which joins CESA only through HVDC links (\S\ref{sec:bg:grid}).
Western Denmark, Latvia, and Bulgaria, in contrast, are missing because of gaps in the extracted lines.
With 1{,}555 stations that no extracted line reaches and the stations in T\"urkiye and North Africa, these areas account for the 5{,}331 substations outside $C^\ast$.
As a result, 79 of the 552 facilities have all their candidates outside $C^\ast$.
To test whether the regions beyond the scope change the result, we rebuild the backbone in two ways.
Removing every link between synchronous areas leaves a component of 18{,}010 substations.
Restricting the graph to the scope before taking the largest component leaves 16{,}248 substations and 27{,}361 corridors.
At a 5\% seed, 52.5\% and 50.6\% of IXP presence go offline on these backbones, against 51.7\% on $C^\ast$ (Table~\ref{tab:sensitivity_full}).
We therefore keep $C^\ast$ as extracted, including the regions beyond the scope.

\begin{table}[!t]
\centering
\small
\caption{IXP presence offline at a 5\% betweenness seed under each model parameter. The baseline is the largest-island rule with $k=3$, a 2.0~km corridor buffer, and an event that outlasts every backup. Each row changes one parameter and keeps the others at the baseline. $^\ast$In parentheses, the value when candidates outside the largest component $C^\ast$ are ignored.}
\label{tab:sensitivity_full}
\begin{tabular}{@{}llr@{}}
\toprule
Parameter & Setting & IXP presence offline (\%) \\
\midrule
Island rule & Direct model, no propagation & 0.1 \\
 & Generation-aware, $\beta$ = 0.5 / 0.75 / 1 & 8.8 / 11.9 / 14.6 \\
 & Largest island (baseline) & 51.7 \\
 & Load redistribution, then largest island & 64.7 \\
\addlinespace
Event duration $T$ & $T \le 12$~h, 12~h of fuel at every facility & 0.0 \\
 & $T = 16$~h, 12~h of fuel at every facility & 51.7 \\
 & $T = 16$~h, a quarter of the facilities hold 48~h & 38.6 \\
\addlinespace
Seed type & 5\% of corridors by edge betweenness & 46.1 \\
Candidate set size & $k$ = 1 to 5 & 49.9--54.7 \\
 & All within 15~km, median 22 & 25.6 (59.8)$^\ast$ \\
Feeder rule & Calibrated single feeder (\S\ref{sec:model:system}) & 55.8 (54.7--56.7) \\
Candidate selection, $k=3$ & Drop nearest, random backup, tiered & 50.0--55.6 \\
Corridor buffer & 0.6 to 2.0~km & 45.9--52.6 \\
Links between synchronous areas & Removed & 52.5 \\
Graph restricted to the scope & Largest component, 16{,}248 substations & 50.6 \\
Island rules in every component & All substations of the graph, same seed size & 51.7 \\
\bottomrule
\end{tabular}
\end{table}

\smallskip\noindent\textbf{Ceiling.}
An island rule de-energizes only substations of $C^\ast$.
Because a facility goes dark only when all its candidate substations are de-energized, a cascade can reach only the facilities whose candidates all lie in $C^\ast$.
For the measured map, these are the 382 facilities with 908 of the 1{,}392 IXP presences, which caps the loss of any seed at 65.2\%.

\begin{proposition}[Ceiling]
\label{prop:ceiling}
For every seed $S_0 \subseteq C^\ast$ and every island rule, $\Phi(X) \le \dfrac{\sum_{f : N(f) \subseteq C^\ast} w_f}{\sum_f w_f}$.
\end{proposition}

\begin{proof}
Because the seed and every island that an island rule de-energizes are subsets of $C^\ast$, the set $X$ of de-energized substations is a subset of $C^\ast$ as well.
A facility is dark only if $N(f) \subseteq X \subseteq C^\ast$.
The numerator of Eq.~\eqref{eq:ixp_offline} is therefore at most the weight of the facilities with $N(f) \subseteq C^\ast$.
\end{proof}

\smallskip\noindent\textbf{Seeds.}
We sweep the seed fraction $q$ from 0 to 40\% in 80 steps along one ranking, so that each seed contains all smaller seeds.
We approximate betweenness from 500 pivot substations.
The random baseline averages 40 orders.
Both island rules are reported in Table~\ref{tab:island_sweep} at the seed fractions of 1, 2, 5, and 10\%, which \name also uses.
Because the splits of 2003 and 2006 began with line trips (Table~\ref{tab:bg:taxonomy}), we also seed with transmission corridors ranked by edge betweenness, which fail without any substation failing directly (Table~\ref{tab:sensitivity_full}).

\begin{table}[!t]
\centering
\caption{IXP presence offline (\%) under the two island rules, with seeds ranked by grid betweenness (substations) or edge betweenness (corridors). The last two columns give the seed fraction at which half of IXP presence is offline, with a dash where it is not reached, and the maximum over seeds up to 40\%. Corridor rows use the monotone envelope at the nearest sweep point, and the last row counts the islands that the generation-aware rule with $\beta=1$ keeps energized.}
\label{tab:island_sweep}
\small
\begin{tabular}{llrrrrrr}
\toprule
 & & \multicolumn{4}{c}{Seed fraction} & & \\
\cmidrule(lr){3-6}
Rule & Seed & 1\% & 2\% & 5\% & 10\% & 50\% at & Max \\
\midrule
Largest island & Substations & 16.9 & 28.4 & 51.7 & 62.4 & 4.6\% & 65.2 \\
Largest island & Corridors & 4.0 & 40.4 & 46.1 & 54.7 & 8.6\% & 63.9 \\
Generation-aware, $\beta=0.5$ & Substations & 0.1 & 1.4 & 8.8 & 17.6 & -- & 45.3 \\
Generation-aware, $\beta=0.75$ & Substations & 0.1 & 1.4 & 11.9 & 19.3 & -- & 48.7 \\
Generation-aware, $\beta=1$ & Substations & 0.1 & 4.9 & 14.6 & 30.2 & 19.2\% & 50.9 \\
\midrule
Islands energized, $\beta=1$ & Substations & 29/117 & 67/241 & 151/646 & 276/1{,}326 & & \\
\bottomrule
\end{tabular}
\end{table}

\smallskip\noindent\textbf{Generation and load.}
The generation-aware rule needs the dispatchable capacity and the load of every island.
We take the power plants of powerplantmatching~\cite{gotzens2019ppm} that operate in 2026 with at least 50~MW of dispatchable capacity, which covers every fuel except wind and solar.
Each plant is assigned to its nearest substation within 5~km.
The load of a substation is the 2019 peak load of its country~\cite{opsd2020timeseries} divided by the number of substations in that country.
Because powerplantmatching lists no plants in Belarus, Russia, and Morocco, we set the load of their 1{,}442 substations to zero.
We evaluate the rule at $\beta = 0.5$, 0.75, and 1 (Table~\ref{tab:island_sweep}).

\smallskip\noindent\textbf{Historical check.}
To compare the island rules with the two documented splits of CESA, we cut the transmission graph along the three areas of 4 November 2006~\cite{ucte2007split} and along the Italian border of 28 September 2003~\cite{ucte2004italy}.
For each area, we compare the capacity and load that the rule assigns with the generation and load reported at the event (Table~\ref{tab:historical_check}).
With the capacity of 2026, every area has $G/L > 1$, and the rule with $\beta = 1$ therefore keeps all four energized.
This outcome matches the three islands of 2006 but misses the collapse of Italy in 2003.

\begin{table}[!t]
\centering
\caption{Historical check of the generation-aware rule. $G$ is the dispatchable capacity of 2026 assigned to an area and $L$ its load proxy, against the generation and load reported just before each split~\cite{ucte2007split,ucte2004italy}. The rule predicts survival when $G/L \geq \beta$. The areas of 2006 follow the operators listed in the final report, with the split inside Germany, Austria, and Croatia approximated by a line through the tripped corridors.}
\label{tab:historical_check}
\small
\begin{tabular}{lrrrrrrl}
\toprule
 & \multicolumn{3}{c}{Rule, capacity of 2026} & \multicolumn{3}{c}{Reported at the event} & \\
\cmidrule(lr){2-4}\cmidrule(lr){5-7}
Island & $G$ (GW) & $L$ (GW) & $G/L$ & Gen. (GW) & Load (GW) & Ratio & Outcome \\
\midrule
2006 West       & 342.4 & 266.2 & 1.29 & 182.7 & 191.6 & 0.95 & Survived, 17~GW shed \\
2006 North-East & 80.1  & 75.7  & 1.06 & 62.3  & 52.3  & 1.19 & Survived \\
2006 South-East & 33.4  & 24.9  & 1.34 & 29.1  & 29.9  & 0.97 & Survived \\
2003 Italy      & 67.7  & 51.9  & 1.30 & 20.5  & 27.4  & 0.75 & Collapsed \\
\bottomrule
\end{tabular}
\end{table}

\smallskip\noindent\textbf{Overload cascade.} The Italian blackout of 2003 and the split of 2006 spread through line overloads after the first trips~\cite{ucte2004italy,ucte2007split}.
We model this mechanism with the load-redistribution cascade of Motter and Lai~\cite{motter2002cascade}, in which the load $L_i(t)$ of substation $i$ is its betweenness in the graph that survives at step $t$, over unweighted shortest paths.
Each substation $i$ has the capacity
\begin{equation}
C_i = (1+\alpha)\,\max\bigl(L_i(0),\,\varepsilon\bigr),
\label{eq:capacity}
\end{equation}
where the tolerance $\alpha \geq 0$ sets the headroom, and $\varepsilon$, the smallest positive initial load, keeps a substation with zero initial load from tripping on any positive load.
The seed $S_0$ fails at step $t=0$, and at each later step every surviving substation with $L_i(t) > C_i$ trips at once.
The cascade stops when no substation trips, or after 50 steps, and we write $Y$ for the set of tripped substations.
Because we estimate the load from 500 source substations, we reuse the same sources at every step, so that no substation trips on sampling noise.
The \emph{overload-only} set $X_{\mathrm{ov}} = S_0 \cup Y$ keeps every surviving substation energized.
The \emph{overload with islanding} set also applies the largest-island rule of Eq.~\eqref{eq:two_stage} to $G \setminus (S_0 \cup Y)$.
We evaluate the tolerances $\alpha \in \{0.1, 0.25, 0.5, 1.0\}$.

At a 5\% seed, 64.7\% of IXP presence goes offline when the largest-island rule follows the overload cascade, close to the ceiling of 65.2\% (Proposition~\ref{prop:ceiling}).

\subsection{Duration and Sensitivity of the Results}
\label{app:sensitivity_feeder}

\smallskip\noindent\textbf{Event duration under backup mixes.}
\label{app:duration}
For each backup mix of Figure~\ref{fig:outage_duration_main}, we assign every facility at random to one of the three backup classes of \S\ref{sec:results:duration} in the proportions of the mix, and we repeat the assignment 200 times.
The lines of Figure~\ref{fig:outage_duration_main} give the mean over these assignments and the bands the 5th--95th percentiles.
For a 2-hour event with half of the facilities on UPS only, the band spans 22.3--29.5\% around a mean of 25.8\%.
When every facility holds UPS only or every facility holds 48~hours of fuel, the loss rises in a single step to the same 51.7\%, at 15~minutes or at 48~hours.

\smallskip\noindent\textbf{Feeder configurations.}
We repeat the most-shared removal of \S\ref{sec:results:bipartite} and the cascade of \S\ref{sec:results:cascade} for eight feeder configurations with at most five candidates.
Five of them take the $k=1$--5 nearest substations.
\emph{Drop nearest} keeps only the second- and third-nearest, and \emph{random backup} keeps the nearest and draws two more at random within 15~km.
\emph{Tiered} replaces the third candidate with the nearest substation at $\geq$220~kV.
A ninth configuration, \emph{all in range}, links a facility to every substation within 15~km, a median of 22.
Under each of the eight configurations, most-shared removal cuts power to half of IXP presence after 13--28\% of the facility-connected substations, against 50--87\% for random removal.
At a 5\% seed, the cascade takes 50--56\% of IXP presence offline under each of them.
Random removal thus needs 2.9--4.9 times as many substations as most-shared removal across the eight configurations.
Even when a facility is linked to every substation within 15~km, most-shared removal needs 39.4\% of the 3{,}874 connected substations against 97.4\% for random removal, 2.5 times as many.

\section{\name in Detail}
\label{app:shield_details}

\subsection{The Hardening Problem}
\label{app:shield_problem}

\label{app:proofs}

To choose what to harden, \name scores every candidate set by the IXP presence that an outage scenario takes offline with this set in place.
An event de-energizes the substations $X$ for $T$ hours.
\name backs up the substations $W$ or extends the backup of the facilities $H$, in both cases to $B_{\mathrm{hard}} = 48$~hours.
Both actions keep the map of \S\ref{section:methodology} fixed, with the same substations, lines, candidate feeders, and IXP facilities.
The substations that stay de-energized are then
\begin{equation}
X_{\mathrm{eff}} =
\begin{cases}
X \setminus W & \text{if } T \le B_{\mathrm{hard}},\\
X & \text{otherwise.}
\end{cases}
\label{eq:shield_eff}
\end{equation}
Facility $f$ is dark when $N(f) \subseteq X_{\mathrm{eff}}$ and offline when it is dark and $T > B_f$, with $B_f = B_{\mathrm{hard}}$ for $f \in H$.
With $\mathrm{Off}(X, W, H)$ the set of offline facilities, the loss is
\begin{equation}
\mathrm{Loss}(X, W, H) = \sum_{f \in \mathrm{Off}(X, W, H)} w_f \;+\; \lambda\, R\bigl(\mathrm{Off}(X, W, H)\bigr),
\label{eq:shield_loss}
\end{equation}
where $R$ is the stranded IXP presence of \S\ref{sec:shield:routing} and $\lambda \ge 0$ its weight.
Against the attacker of T1, the defender solves
\begin{equation}
\min_{|W| \le B_d} \;\; \max_{|X| \le B_a} \;\; \mathrm{Loss}(X, W, \varnothing),
\label{eq:shield_objective}
\end{equation}
or the analogue over $H$ for facility backup.
For T2 and T3, the defender minimizes the mean loss over the fixed scenarios $X_1,\dots,X_n$.

When every unhardened facility has the same backup, both kinds of hardening help only for events that outlast this backup but not the 48~hours of the hardened units.
The evaluation of \S\ref{sec:shield:results} therefore fixes the event at 16~hours.

\begin{proposition}[Duration window]
\label{prop:window}
Let $B_f = b$ for every facility that is not hardened, and let $B_{\mathrm{hard}} > b$.
The loss is zero for $T \le b$.
For $T > B_{\mathrm{hard}}$ the loss equals the loss without hardening.
Hardening lowers the loss only for $b < T \le B_{\mathrm{hard}}$.
\end{proposition}

\begin{proof}
A facility is offline only when $T > B_f$.
For $T \le b$, no facility satisfies this condition, because $B_f \ge b$ for every facility.
The set of offline facilities is then empty, no metropolitan area is down, $R(\varnothing) = 0$, and the loss is zero.
For $T > B_{\mathrm{hard}}$, Eq.~\eqref{eq:shield_eff} gives $X_{\mathrm{eff}} = X$, and $T > B_f$ holds for every facility, hardened or not.
The set of offline facilities is then the set of facilities with $N(f) \subseteq X$, which does not depend on $W$ or $H$.
\end{proof}

Against fixed scenarios, as in T2 and T3, choosing the substations to back up is a coverage problem, for which a greedy defender has a guarantee.

\begin{proposition}[Scenario hardening is a coverage problem]
\label{prop:coverage}
Fix scenarios $X_1,\dots,X_n$, $\lambda=0$, and $b < T \le B_{\mathrm{hard}}$.
The protection of substation backup, i.e., the mean loss without hardening minus the mean loss with $W$, is a monotone submodular function of $W$.
The greedy algorithm reaches at least a share $1-1/e$ of the optimal protection for every budget~\cite{nemhauser1978analysis}.
The protection of facility backup is a modular function of $H$, and the greedy algorithm is optimal.
\end{proposition}

\begin{proof}
For $b < T \le B_{\mathrm{hard}}$ and substation backup, facility $f$ is offline in scenario $X_i$ when $N(f) \subseteq X_i \setminus W$, which holds when $N(f) \subseteq X_i$ and $N(f) \cap W = \varnothing$.
The protection of $W$ is therefore
\[
P(W) = \frac{1}{n} \sum_{i=1}^{n} \; \sum_{f \,:\, N(f) \subseteq X_i} w_f \, \mathbf{1}[N(f) \cap W \neq \varnothing].
\]
Each indicator is the coverage of the element $(i,f)$ by the sets $\{(i,f) \mid s \in N(f)\}$ of the substations $s \in W$.
$P$ is a non-negative weighted sum of coverage functions, hence monotone and submodular, and the greedy algorithm reaches a share $1-1/e$ of the optimum under a cardinality constraint~\cite{nemhauser1978analysis}.
For facility backup, a hardened facility is never offline and the other facilities are unaffected, which gives $P(H) = \sum_{f \in H} w_f \, |\{i \mid N(f) \subseteq X_i\}| / n$.
This function is modular, and the $B_d$ facilities with the largest terms are optimal.
\end{proof}

Against the attacker of T1, this guarantee does not hold, because the T1 objective is not submodular in $W$.
We therefore measure the gap of the greedy defender instead (Appendix~\ref{sec:shield:solver}).
For example, take two facilities $f_1$ and $f_2$ of unit weight with disjoint candidate sets of size $k$, $B_a = k$, $\lambda = 0$, and $b < T \le B_{\mathrm{hard}}$.
Because the attacker can darken one unhardened facility, the worst-case loss is 1 for $W = \varnothing$, for $W = N(f_1)$, and for $W = N(f_2)$, and 0 for $W = N(f_1) \cup N(f_2)$.
The protection $L(\varnothing) - L(W)$ is 0, 0, 0, and 1 for these four sets.
Adding $N(f_2)$ to $N(f_1)$ gains 1 and adding $N(f_2)$ to $\varnothing$ gains 0, which violates the diminishing-returns condition of submodularity.

The attacker problem of T1 is NP-hard, but its solution limits where the defender searches.

\begin{proposition}[Attacker and pruning]
\label{prop:attacker}
For $k \ge 2$ the attacker problem of Eq.~\eqref{eq:shield_objective} is NP-hard, and so is the defender problem.
Let $\lambda=0$, let $L(W)$ be the worst-case loss against $W$, and let $X^\ast$ be an optimal attack against $W$.
For every $W' \supseteq W$ with $(W' \setminus W) \cap X^\ast = \varnothing$, $L(W') = L(W)$.
\end{proposition}

\begin{proof}
For hardness, we reduce the densest-subgraph problem to the attacker problem.
Given a graph $G$ and an integer $m$, this problem asks for the $m$ vertices that span the most edges.
The densest-subgraph problem is NP-hard, because $G$ has a clique of $m$ vertices exactly when its densest subgraph on $m$ vertices spans $m(m-1)/2$ edges~\cite{feige2001dense}.
From $G$ we build an instance with $k=2$, one substation per vertex, and one facility of unit weight per edge, whose two candidates are the endpoints of the edge.
We set $B_a = m$, $W = \varnothing$, $\lambda = 0$, and $b < T \le B_{\mathrm{hard}}$.
A facility then goes offline exactly when the attack contains both endpoints of its edge, and the loss of an attack is therefore the number of edges that its substations span.
An optimal attack thus selects a densest subgraph on $m$ vertices, and the instance is built in time linear in the size of $G$.
For $k > 2$, we add the same $k-2$ new substations to every candidate set and raise $B_a$ by $k-2$.
Because every attack that takes a facility offline contains the new substations, the optimal attacks are those of the case $k=2$ plus these substations.
Because the defender problem of Eq.~\eqref{eq:shield_objective} with $B_d = 0$ is the attacker problem, it is NP-hard as well.
For the second statement, hardening is monotone, because $X \setminus W' \subseteq X \setminus W$ for $W' \supseteq W$, and therefore $L(W') \le L(W)$.
If $(W' \setminus W) \cap X^\ast = \varnothing$, then $X^\ast \setminus W' = X^\ast \setminus W$.
The attack $X^\ast$ therefore causes the same loss against $W'$ as against $W$, which gives $L(W') \ge L(W)$.
\end{proof}
By the second statement, the greedy defender of \name needs to search only the substations of the current optimal attack (Appendix~\ref{sec:shield:solver}).

\subsection{Solver}
\label{sec:shield:solver}
\smallskip\noindent\textbf{Exact attacker.}
Because a greedy attacker underestimates the loss (below), \name solves the attacker problem of T1 exactly, as a mixed-integer linear program.
A binary variable $x_s$ marks an attacked substation and a variable $y_f$ a dark facility, under the constraints $y_f \le x_s$ for every $s \in N(f)$ and $\sum_s x_s \le B_a$.
The program maximizes $\sum_f w_f y_f$ over the facilities that can go offline, with $T > B_f$ and no hardened substation in $N(f)$.
HiGHS~\cite{huangfu2018highs} solves each to proven optimality.

\smallskip\noindent\textbf{Greedy defender.}
Because in T1 only a unit that the current optimal attack uses can lower the worst-case loss (Proposition~\ref{prop:attacker}), the candidates of each greedy step are the substations of the facilities that this attack takes offline.
The search thus covers all 984 substations without a shortlist.
If no candidate lowers the loss, \name breaks the tie by hardening the candidate that serves the most IXP presence at these facilities.

\begin{table}[!t]
\centering
\small
\caption{Gap of the greedy defender of \name against the exact defender, for substation backup unless stated, a 16-hour event, and $\lambda = 0$. Gaps are in IXP presences, with the largest absolute gap per row.}
\label{tab:shield_gap}
\begin{tabular}{llll}
\toprule
Threat & Exact for & Greedy optimal & Largest gap \\
\midrule
T1, $B_a = 5$ & $B_d \le 5$ & $B_d = 1$ & 11 of 41 ($B_d$=2) \\
T1, $B_a = 10$ & $B_d \le 5$ & all & 0 \\
T1, $B_a = 20$ & $B_d \le 4$ & all & 0 \\
T1, facility backup & $B_d \le 3$ & all & 0 \\
T2 & $B_d \le 30$ & $B_d \le 6$ & 5.5 of 236 ($B_d$=27) \\
T3 & $B_d \le 30$ & all & 0 \\
\bottomrule
\end{tabular}
\end{table}

\smallskip\noindent\textbf{Greedy attackers underestimate the loss.}
A common greedy attacker adds at each step the substation with the largest marginal loss.
With three candidates per facility, the first substation that this attacker adds darkens no facility, and the attacker therefore starts without a signal.
Without hardening, the greedy attacker takes 4, 111, and 227 IXP presences offline for $B_a = 5$, 10, and 20, against 77, 167, and 271 for the exact attacker.
This attacker thus underestimates the loss by 94.8\%, 33.5\%, and 16.2\%.
At $B_a = B_d = 10$, \name leaves 47 IXP presences offline under the exact attacker, against 136 for a defender that plans against the greedy attacker.

\smallskip\noindent\textbf{Exact defender.}
To measure the gap of the greedy defender, we compare it with an exact defender.
In T1, the exact defender uses implicit enumeration, because extending a hardening set $W$ lowers its worst-case loss only if the extension contains a substation of the optimal attack against $W$ (Proposition~\ref{prop:attacker}).
The search therefore branches only on those substations and solves the exact attacker at every node, which yields the optimum over all subsets of the 984 substations and not over a shortlist.
In T2 and T3, we solve the coverage problem as a mixed-integer linear program for every budget from 1 to 30.
The greedy defender is optimal at every budget in T3, and its largest gaps are 11 of 41 IXP presences in T1 with $B_a = 5$ and 5.5 of 236 (2.3\%) in T2 (Table~\ref{tab:shield_gap}).
We did not measure the gap for $\lambda > 0$ in T2 and T3, where the objective has a threshold and the coverage guarantee of Proposition~\ref{prop:coverage} does not apply.

\subsection{Additional Results}
\label{app:shield_more}
\smallskip\noindent\textbf{Rankings and budgets.}
We compare \name with the best of the five rankings of \S\ref{sec:shield:results} at budgets of 5, 10, and 20 (Table~\ref{tab:shield_headline_full}).
The threats are the attackers on 5, 10, and 20 substations and the two cascades.
The most-shared ranking here breaks ties by OSM identifier rather than by IXP presence served, which changes one of its first ten substations against the order of \S\ref{sec:results:bipartite}.
Against the attacker on ten substations, backing up the ten substations of \name leaves no IXP down and 21 degraded, against 5 down and 46 degraded without hardening.

\smallskip\noindent\textbf{Stranded traffic.}
At the 5\% seed of the largest-island rule (Figure~\ref{fig:shield_failover}), Amsterdam receives 40.3\% of the monitored addresses of Warsaw and 24.0\% of those of Paris.
Paris receives 17.2\% of Madrid's and 6.9\% of Amsterdam's.
Of the 45.6 stranded IXP presences, 21.0 come from Warsaw.

\begin{figure}[!t]
\centering
\begin{minipage}[t]{0.57\linewidth}
\vspace{0pt}
\centering
\small
\setlength{\tabcolsep}{3pt}
\captionof{table}{IXP presence offline (count, of 1{,}392) for a 16-hour event, substation backup, and $\lambda = 0$. T2 and T3 report the mean over four scenarios. The best ranking is the best of the five rankings at that budget, chosen after the fact.}
\label{tab:shield_headline_full}
\begin{tabular}{lrrrrrrr}
\toprule
 & & \multicolumn{3}{c}{\name} & \multicolumn{3}{c}{Best ranking} \\
\cmidrule(lr){3-5}\cmidrule(lr){6-8}
Threat & $B_d = 0$ & 5 & 10 & 20 & 5 & 10 & 20 \\
\midrule
T1, $B_a = 5$ & 77 & 29 & 20 & 16 & 41 & 41 & 28 \\
T1, $B_a = 10$ & 167 & 69 & 47 & 34 & 90 & 80 & 58 \\
T1, $B_a = 20$ & 271 & 119 & 94 & 68 & 150 & 140 & 107 \\
T2 & 554.3 & 421.0 & 361.8 & 285.8 & 468.5 & 460.5 & 423.3 \\
T3 & 173.0 & 133.5 & 109.5 & 80.5 & 162.3 & 162.3 & 154.0 \\
\bottomrule
\end{tabular}
\end{minipage}\hfill
\begin{minipage}[t]{0.40\linewidth}
\vspace{0pt}
\centering
\includegraphics[width=\linewidth]{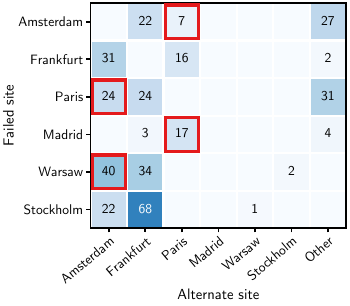}
\captionof{figure}{Share (\%) of the monitored addresses of each failed PEERING site (\S\ref{section:routing_results}) that reappears at each of the other five European sites or at a site outside the scope. Red frames mark the pairs of sites that go down together at the 5\% seed of the largest-island rule, whose share is therefore stranded.}
\label{fig:shield_failover}
\end{minipage}
\end{figure}

\smallskip\noindent\textbf{Hedged set.}
We compute the hedged set of \S\ref{sec:shield:robustness} greedily under the $k = 3$ feeder model.

\smallskip\noindent\textbf{Heterogeneous backup.}
To test whether \name needs to know the backup of each facility, we draw the backup assignment of \S\ref{sec:shield:duration} 20 times.
At $B_d = 10$ and for a 16-hour event, the set chosen with knowledge of each draw leaves a mean of 41.2, 260.7, and 77.5 IXP presences offline under T1, T2, and T3.
The set chosen for 12~hours at every facility leaves 45.2, 270.7, and 82.2, only 4--10 more.
Under this assignment, IXP presence goes offline after 15~minutes.
For a 2-hour event, backing up the ten substations of \name lowers the loss from 7.5\% to 2.9\% under T1, from 20.2\% to 12.9\% under T2, and from 6.1\% to 3.8\% under T3.

\begin{figure}[!tb]
\centering
\includegraphics[width=0.85\linewidth]{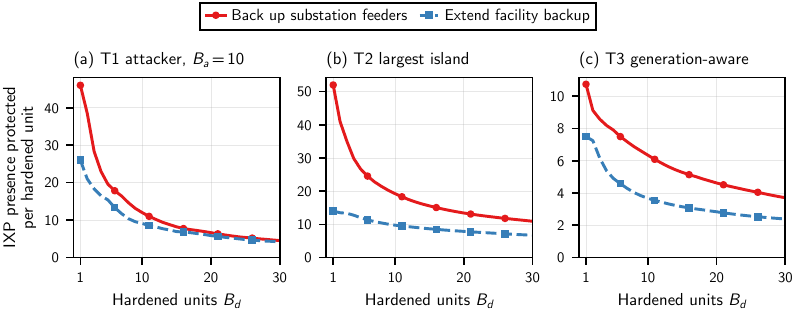}
\caption{IXP presence protected per hardened unit when \name backs up substation feeders or extends facility backup, for a 16-hour event and $\lambda = 0$. Substation backup protects more per unit, because a backed-up substation covers every facility that has it as a candidate. The advantage is largest under T2 (b) and shrinks with the budget under T1 (a).}
\label{fig:shield_levers}
\end{figure}

\begin{table}[!tb]
\centering
\footnotesize
\caption{IXP presences offline for a cost budget of ten when a substation costs $c$ times a facility, for a 16-hour event and $\lambda = 0$, with T2 and T3 as the mean over four seeds. ``Both'' lets the defender pick substations or facilities by protection per cost, with the substations plus facilities it picks in parentheses. Without hardening the loss is 167, 554.3, and 173.0.}
\label{tab:shield_cost}
\setlength{\tabcolsep}{4pt}
\begin{tabular}{@{}llrrrr@{}}
\toprule
Threat & Hardening & $c = 1$ & $c = 2$ & $c = 3$ & $c = 10$ \\
\midrule
T1, $B_a = 10$ & Substations only & 47 & 69 & 82 & 121 \\
 & Facilities only & 76 & 76 & 76 & 76 \\
 & Both, by cost & 47 (10+0) & 67 (1+8) & 73 (0+10) & 73 (0+10) \\
\addlinespace
T2 & Substations only & 361.8 & 421.0 & 449.2 & 502.2 \\
 & Facilities only & 456.8 & 456.8 & 456.8 & 456.8 \\
 & Both, by cost & 361.8 (10+0) & 412.8 (3+4) & 431.8 (1+7) & 456.8 (0+10) \\
\addlinespace
T3 & Substations only & 109.5 & 133.5 & 147.2 & 162.2 \\
 & Facilities only & 136.2 & 136.2 & 136.2 & 136.2 \\
 & Both, by cost & 109.5 (10+0) & 127.0 (3+4) & 132.2 (1+7) & 136.2 (0+10) \\
\bottomrule
\end{tabular}
\end{table}

\begin{table}[!tb]
\centering
\small
\caption{Stranded IXP presence without hardening, with five, and with ten backed-up substations of \name chosen without stranded traffic in the objective, for a 16-hour event and a metro that counts as down at the given share of its IXP presence offline. The T1 column is the attacker on twenty substations who targets stranded traffic. The last column gives how many of the first five substations change when \name counts stranded traffic in its objective, under T1 with ten and twenty attacked substations, T2, and T3.}
\label{tab:shield_threshold}
\setlength{\tabcolsep}{4pt}
\begin{tabular}{@{}lccccc@{}}
\toprule
Down at & T2, 5\% seed & T2, 10\% seed & T3, 10\% seed & T1, $B_a = 20$ & First five changed \\
\midrule
Any offline & 45.6 / 45.6 / 45.6 & 182.1 / 182.1 / 182.1 & 40.7 / 40.7 / 40.7 & 83.0 / 121.8 / 11.5 & 2, 1, 1, 1 \\
25\% & 45.6 / 40.6 / 0 & 182.1 / 176.4 / 0 & 40.7 / 5.0 / 0 & 224.3 / 89.3 / 0 & 0, 2, 2, 0 \\
50\% & 45.6 / 0 / 0 & 45.6 / 0 / 0 & 5.0 / 0 / 0 & 121.8 / 0 / 0 & 0, 1, 0, 0 \\
75\% & 40.6 / 0 / 0 & 40.6 / 0 / 0 & 0 / 0 / 0 & 83.0 / 0 / 0 & 0, 0, 0, 0 \\
100\% & 40.6 / 0 / 0 & 40.6 / 0 / 0 & 0 / 0 / 0 & 0 / 0 / 0 & 0, 0, 0, 0 \\
\bottomrule
\end{tabular}
\end{table}

\end{document}